\documentclass[a4paper,
11pt,openright]{article} 
\usepackage[english]{babel} 
\usepackage[latin1]{inputenc} 
\usepackage{amsmath,amssymb,amsthm,amsfonts,amscd,authblk,dsfont,color,cancel,diagbox,enumerate,epsfig,esint,eucal,fancyhdr,latexsym,lmodern,mathabx,mathrsfs,mathtools,multicol,multirow,musicography,phaistos,skull,simpler-wick,subcaption,tikz-cd,ulem,xfrac}
\usepackage[multiuser]{fixme}
\FXRegisterAuthor{nd}{annd}{Nico} 
\FXRegisterAuthor{cvdv}{ancvdv}{Chris}
\FXRegisterAuthor{lp}{anlp}{Lo}
\usepackage[
final=true,                    
plainpages=false,              
pdfpagelabels=true,            
pdfencoding=auto,              
unicode=true,                  
hypertexnames=true,            
naturalnames=true              
]{hyperref}
\usepackage[all,cmtip]{xy}
\usetikzlibrary{matrix,calc}
\usepackage[autostyle, english = american]{csquotes}
\MakeOuterQuote{"} 

\definecolor{NiColor}{RGB}{77,77,255}
\definecolor{NiColoRed}{RGB}{255,77,77}
\definecolor{NiCitation}{RGB}{0,181,26}

\newtheoremstyle{TheoremStyle}
{\topsep}
{\topsep}
{}
{}
{\sc}
{:}
{.5em}
{}
\makeatletter
\def\@endtheorem{\begin{flushright}$\diamond$\end{flushright}}
\makeatother

\theoremstyle{TheoremStyle}
\newtheorem{theorem}{Theorem}[section]
\newtheorem{corollary}[theorem]{Corollary}

\newtheorem{proposition}[theorem]{Proposition}
\newtheorem{lemma}[theorem]{Lemma}

\newtheorem{remark}[theorem]{Remark}
\newtheorem{example}[theorem]{Example} 

\title{Subcriticality at High Temperatures in Spin Lattice Systems}
\author[a]{N Drago\thanks{\href{mailto:nicolo.drago@unige.it}{nicolo.drago@unige.it}}}
\author[b]{L. Pettinari\thanks{\href{mailto:lorenzo.pettinari@unitn.it}{lorenzo.pettinari@unitn.it}}}
\author[c]{C. J. F. van de Ven\thanks{\href{mailto:chris.ven@fau.de}{chris.ven@fau.de}}}

\affil[a]{Dipartimento di Matematica, Universit\`{a} di Genova and INdAM, Via Dodecaneso 35, I-16146 Genova, Italy}
\affil[b]{Dipartimento di Matematica, Universit\`a di Trento and INFN-TIFPA and INdAM, Via Sommarive 14, I-38123 Povo, Italy}
\affil[c]{Friedrich-Alexander-Universit\"at Erlangen-N\"urnberg, Department of Mathematics,
Cauerstra\ss e 11 91058 Erlangen, Germany}

\begin{document}
\maketitle

\begin{abstract}
We provide new sufficient conditions for subcriticality of classical and quantum spin lattice systems, formulated in terms of the uniqueness of Kubo-Martin-Schwinger (KMS) states.
This is achieved by exploiting a non-commutative analog of the Kirkwood-Salzburg equations together with a novel decomposition of local observables.
In contrast to standard approaches \cite{Bratteli_Robinson_97,Frohlich_Ueltschi_2015}, our condition is uniform with respect to the dimension of the single-site Hilbert space. Moreover, unlike the results of \cite{Drago_Pettinari_Van_de_Ven_2025}, which required control over the growth of the derivatives of the interaction potentials, our result only involves estimating the natural $C^*$-norm of these potentials. This substantially enlarges the class of interactions for which the theorems apply and provides better lower bounds on the subcritical inverse temperature.
\end{abstract}


\section{Introduction}
\label{Sec: introduction}

In the algebraic approach to classical and quantum theories, the Kubo-Martin-Schwinger (KMS) condition has been developed to identify states describing thermal equilibrium \cite{Bratteli_Robinson_97, Haag_Hugenholtz_Winnik_1967,Israel_1979, Pusz_Woronowicz_1978}.

For a quantum system described by a non-commutative $C^*$-algebra $\mathfrak{A}$, whose self-adjoint elements represent physical observables, the KMS condition is summarized as follows \cite{Bratteli_Robinson_97}.
Given a strongly continuous one-parameter group $\tau$ of $C^*$-automorphisms on $\mathfrak{A}$, which is interpreted as the dynamics on the underlying physical system, a state $\omega\in S(\mathfrak{A})$ is called \textbf{$(\beta,\tau)$-KMS quantum state} if
\begin{align}\label{Eq: quantum KMS condition}
	\omega(\mathfrak{a}\tau_{i\beta}\mathfrak{b})
	=\omega(\mathfrak{b}\mathfrak{a})\,,
\end{align}
for all $\tau$-analytic elements $\mathfrak{a},\mathfrak{b}\in\mathfrak{A}_\tau$.
Here $\beta>0$ represents an inverse temperature.
For a classical system described by a $C^*$-Poisson algebra $\mathscr{A}$, the KMS condition is instead formulated with the help of the Poisson bracket of $\mathscr{A}$ \cite{Aizenman_Gallavotti_Goldstein_Lebowitz_1976,Aizenman_Goldstein_Gruber_Lebowitz_Martin_1977,Drago_van_de_Ven_2023,Drago_Waldmann_2024,Gallavotti_Pulvirenti_1976,Gallavotti_Verboven_1975}.
In particular, given a derivation $\delta\colon\dot{\mathscr{A}}\to\mathscr{A}$ on $\mathscr{A}$, defined on the dense domain $\dot{\mathscr{A}}$ of the Poisson bracket $\{\;,\;\}$, a state $\omega\in S(\mathscr{A})$ is called a \textbf{$(\beta,\delta)$-KMS classical state} if
\begin{align}\label{Eq: classical KMS condition}
	\omega(\{a,b\})
	=\beta\omega(b\delta(a))\,,
\end{align}
for all $a,b\in\dot{\mathscr{A}}$.
Within this setting $\delta$ is thought as the infinitesimal generator of the classical dynamics on $\mathscr{A}$.
As shown in \cite{Aizenman_Goldstein_Gruber_Lebowitz_Martin_1977,
Drago_van_de_Ven_2023}, condition \eqref{Eq: classical KMS condition} is equivalent to the Dobrushin-Landford-Ruelle (DLR) condition \cite{Dobrushin_1968_1,Dobrushin_1968_2,Dobrushin_1968_3,Dobrushin_1969,Dobrushin_1970,Lanford_Ruelle_1969,Ruelle_1967} for a large class of physically interesting systems.

A class of models where KMS states have been largely studied is the one of quantum spin lattice systems, where the $C^*$-algebra $\mathfrak{A}$ coincides with the quasi-local $C^*$-algebra $\mathcal{A}_\Gamma$ over a countable set $\Gamma$.
In particular $\mathcal{A}_\Gamma$ is the $C^*$-direct limit of the net of $C^*$-algebras $\{\mathcal{A}_\Lambda\}_{\Lambda\Subset\Gamma}$ associated with finite regions of $\Gamma$.
The associated dynamics $\tau^\Gamma$ is generated by an appropriate family of quantum potential $\Phi=\{\Phi_\Lambda\}_{\Lambda\Subset\Gamma}$ ---\textit{cf.} Section \ref{Subsec: quantum quasi-local algebra} for further details.
These models aim at studying the collecting behaviour of quantum systems located at sites $x\in\Gamma$ with associated algebra of observables $\mathcal{A}_x=\mathcal{B}(\mathcal{H}_x)$, $\mathcal{H}_x$ being a finite dimensional Hilbert space.
The interaction between various sites is described by $\Phi$.
It is customary to split the family of quantum potential $\Phi$ in two contributions, namely one recollecting single-site potentials $\Psi=\{\Psi_x\}_{x\in\Gamma}$, and one describing purely "multilocal" potentials $\bar{\Phi}=\{\bar{\Phi}_\Lambda\}_{\Lambda\Subset\Gamma}$ ---that is, those describing interactions between regions $|\Lambda|>2$, \textit{cf.} Section \ref{Subsec: quantum quasi-local algebra}.

In this paper we will study KMS states on quantum spin lattice systems by considering in particular the \textbf{subcritical} regime.
By definition, an inverse temperature $\beta_{\textsc{u}}$ is called \textbf{subcritical} if the set of KMS states at inverse temperature $\beta$ is a singleton for all $\beta<\beta_{\textsc{u}}$.
This parameter should not be confused with the critical inverse temperature, which instead signals the onset of a phase transition.
Our goal is to provide a sufficient condition on the multi-local potential $\bar{\Phi}$ ensuring the existence of a non-vanishing subcritical inverse temperature.
In other words, we aim at proving uniqueness of KMS states for sufficiently small $\beta$, \textit{i.e.}, at high enough temperatures.

This topic has received a lot of attention \cite{Araki_1975,Kishimoto_1976, Sakai_1976}: Two classical general results can be found in \cite[Prop. 6.2.45, Thm. 6.2.46]{Bratteli_Robinson_97}, see also \cite{Frohlich_Ueltschi_2015}.
These are based on assumptions on the family $\Phi$ (or $\bar{\Phi}$) which essentially constrains the magnitude of $\|\Phi_X\|$ as the size of $X\Subset\Gamma$ grows to infinity.
Specifically, \cite[Prop. 6.2.45]{Bratteli_Robinson_97} and \cite{Frohlich_Ueltschi_2015} obtain uniqueness of KMS quantum states provided a suitably defined norm of the whole family of potential $\Phi$ ---including the single-site potential--- is finite.
In particular, \cite[Prop. 6.2.45]{Bratteli_Robinson_97} assumes that
\[
\exists\varepsilon>0\mid
\|\Phi\|_{\varepsilon+2\log(d+1)}<\infty\,,
\]	
where
\begin{align}\label{Eq: assumption on potential - norm on potentials}
	\|\Phi\|_{\lambda}
	:=\sup_{x\in\Gamma}\sum_{X\ni x}
	e^{\lambda(|X|-1)}\|\Phi_X\|
	<+\infty\,.
\end{align}
Here $d:=\dim\mathcal{H}_x$ is constant for all $x\in\Gamma$.
Conversely, \cite[Thm. 6.2.46]{Bratteli_Robinson_97} generalizes \cite[Prop. 6.2.45]{Bratteli_Robinson_97} by requiring only the finiteness of $\|\bar{\Phi}\|_{\varepsilon+\log2+3\log(d+1)}$ within the assumption that the latter commutes with the single-site potential $\Psi$.

An inherent limitation of the previous results lies in their lack of uniformity with respect to the dimension of the single-site Hilbert space $\mathcal{H}_x$.
In particular, the assumptions stated above are \textit{not} uniform in $d$.
This also extends  to the estimate of the subcritical inverse temperature (below which uniqueness is ensured) which gets lower and lower as $d$ raises.
This is ultimately a limitation preventing the extension of these results to the case of a quantum spin lattice system with infinite dimensional single-site Hilbert space $\mathcal{H}_x$.
For such systems, the only result available, to the best of our knowledge, was obtained in \cite{Albeverio_Kondratiev_Rockner_Tsikalenko_1997}, where a specific class of quantum spin lattice systems was considered.

Recently, a uniqueness result similar in spirit to those of \cite{Bratteli_Robinson_97,Frohlich_Ueltschi_2015} was obtained in \cite{Drago_Pettinari_Van_de_Ven_2025}.
Therein it was shown that uniqueness of both KMS quantum and classical states can be obtained under a condition uniform over the dimension of the single-site Hilbert space.
Despite this achievement, there are a few drawbacks in this approach which we aim to solve in this paper and which we list below.

First of all, the results in \cite{Drago_Pettinari_Van_de_Ven_2025} apply under the assumption that the family of quantum potentials $\Phi$ can be obtained by a  strict deformation quantization \cite{Berezin_1975,Landsman_1998_I,Landsman_2017,Rieffel_94} of a corresponding family of classical potential $\varphi=\{\varphi_\Lambda\}_{\Lambda\Subset\Gamma}$ associated with an $\mathbb{S}^2$-valued classical spin lattice system, that is, $\Phi_\Lambda=Q_\Gamma(\varphi_\Lambda)$ for all $\Lambda\Subset\Gamma$ ---\textit{cf.} Section \ref{Subsec: Uniqueness result for classical KMS states} for further details.
While this assumption is not very restrictive in applications, it is conceptually undesirable.

Secondly, uniqueness of KMS quantum and classical states is proved by considering a condition on the family $\varphi=\{\varphi_\Lambda\}_{\Lambda\Subset\Gamma}$ of classical potentials which involves higher-order derivatives of the potentials.
While this may appear to be the natural trade-off for obtaining results that are uniform with respect to the dimension of the single-site Hilbert space, we aim to show that it is nevertheless possible to achieve the same uniqueness results under assumptions relying solely on the standard $C^*$-norm of the underlying $C^*$-algebra, without invoking differentiability of the involved classical potentials.
In fact, the latter is not required in the proof of uniqueness of KMS classical states based on cluster expansion, \textit{cf.} \cite[Thm. 6.38]{Friedli_Velenik_2017}.
However, it is less clear how to generalize these results to the quantum setting.

Finally, the proofs of \cite[Thm. 4.1-4.7]{Drago_Pettinari_Van_de_Ven_2025} are not sufficiently robust to extend to the case where the main assumptions are imposed only on the multilocal potential $\bar{\Phi}$, without further assumptions on the single-site potential $\Psi$.
In other words, while \cite[Thm 4.7]{Drago_Pettinari_Van_de_Ven_2025} generalizes \cite[Prop. 6.2.45]{Bratteli_Robinson_97}, there exists no analogous generalization of \cite[Thm. 6.2.46]{Bratteli_Robinson_97}. 

Taking this into account, the first main result of this paper is the following (we refer to Section \ref{Sec: Uniqueness result for quantum KMS states} for all technical details):

\begin{theorem}\label{Thm: uniqueness for KMS quantum states - with commutation}
	Assume that the single site potential $\Psi$ commutes with $\bar{\Phi}$, \textit{cf.} condition \eqref{Eq: assumption on potential - commutation}.
	Then if $\beta>0$ fulfils
	\begin{align}\label{Eq: uniqueness for KMS quantum states - condition for uniqueness - with commutation}
		\exists\varepsilon>0\mid
		\beta\|\bar{\Phi}\|_{\varepsilon+\log 3}
		<\frac{1}{6}\frac{\varepsilon}{1+e^\varepsilon}\,.
	\end{align}
	Then there exists a unique $(\beta,\tau^\Gamma)$-KMS quantum state on $\mathcal{A}_\Gamma$.
\end{theorem}
Theorem \ref{Thm: uniqueness for KMS quantum states - with commutation} establishes a novel uniqueness result for KMS quantum states on $\mathcal{A}_\Gamma$.
Our main assumption on $\bar{\Phi}$ consists in a bound on the norm $\|\Bar{\Phi}\|_{\varepsilon+\log 3}$, \textit{cf.} equation \eqref{Eq: uniqueness for KMS quantum states - condition for uniqueness - with commutation}.
In comparison with \cite{Bratteli_Robinson_97,Frohlich_Ueltschi_2015}, this norm is uniform with respect to the dimension of the single-site Hilbert space. Consequently, the same uniformity carries over to our subcritical inverse temperature, \textit{cf.} Remark \ref{Rmk: subcritical inverse temperature}, while simultaneously avoiding the limitations of the analogous result in \cite{Drago_Pettinari_Van_de_Ven_2025}, where the potential were required to be quantization of sufficiently regular, classical interactions.

The proof of Theorem \ref{Thm: uniqueness for KMS quantum states - with commutation} is sufficiently flexible to be applied also in the classical setting.
In particular, in Section \ref{Sec: consequences of (main) Theorem} we will prove the following result.

\begin{theorem}\label{Thm: uniqueness for KMS classical states}
	Let $\varphi=\{\varphi_\Lambda\}_{\Lambda\Subset\Gamma}$ be such that condition \eqref{Eq: assumption on potential - C1 norm on potentials} holds and that
	\[
	\|\bar{\varphi}\|_{\log 3}<+\infty\,.
	\]
	Then there exists a unique $(\beta,\delta_\Gamma)$-KMS classical state for $\beta\in[0,\tilde{\beta}_{\textsc{u}})$ where
	\begin{align}\label{Eq: uniqueness for KMS classical states - estimate of critical temperature}
	\tilde{\beta}_{\textsc{u}}
	:=\frac{\log 2}{6\|\bar{\varphi}\|_{\log 3}}\,.
	\end{align}
\end{theorem}

Theorems \ref{Thm: uniqueness for KMS quantum states - with commutation}--\ref{Thm: uniqueness for KMS classical states} naturally motivate the search for a unified uniqueness result for both classical and quantum KMS states within the framework introduced in \cite{Drago_Pettinari_Van_de_Ven_2025}.
In that setting, a family of classical potentials $\varphi$ is quantized into a family of quantum potentials $\Phi=Q_\Gamma(\varphi)$ by means of a specific quantization map, \textit{cf.} Section \ref{Subsec: common subcritical region for quantum and classical systems}.
This raises the question of whether one can identify sufficient conditions on $\varphi$ that guarantee uniqueness of both classical and quantum KMS states below a common subcritical inverse temperature.
A partial positive answer was obtained in \cite[Thm. 4.7]{Drago_Pettinari_Van_de_Ven_2025}, albeit under a rather strong regularity assumption on $\varphi$ and without disentangling the contribution of the single-site potential $\psi$.

Since Theorem \ref{Thm: uniqueness for KMS classical states} partially addresses this issue, it is natural to ask whether a result analogous to \cite[Thm. 4.7]{Drago_Pettinari_Van_de_Ven_2025} can be established in the present setting.
However, a technical obstruction prevents a direct application of Theorems \ref{Thm: uniqueness for KMS quantum states - with commutation}--\ref{Thm: uniqueness for KMS classical states} to this problem.
Indeed, Theorem \ref{Thm: uniqueness for KMS quantum states - with commutation} assumes that the single-site potential $\Psi$ commutes with the multilocal potential $\bar{\Phi}$.
While this assumption is natural in the quantum setting, it is not easily translated into a condition at the classical level.

More precisely, one would need a condition on the classical single-site potential $\psi$ ensuring that its quantization $\Psi=Q_\Gamma(\psi)$ commutes with the quantization $\bar{\Phi}=Q_\Gamma(\bar{\varphi})$ of the multilocal potential $\bar{\varphi}$, which appears difficult to formulate in a satisfactory way.
To overcome this issue, we extend Theorem \ref{Thm: uniqueness for KMS quantum states - with commutation} to the case in which the single-site potential does not commute with the multilocal potential $\bar{\Phi}$.
This extension requires a stronger assumption on the latter.

Specifically, let $\varepsilon, \zeta > 0$ and define
\begin{align}\label{Eq: assumption on potential - norm on potentials strengthened}
\|\bar{\Phi}\|_{\varepsilon,\zeta}
:= \sup_{x\in\Gamma}
\sum_{\substack{\Lambda \Subset \Gamma \\ x \in \Lambda}}
e^{\varepsilon(|\Lambda|-1) + \zeta \|\Psi\|_\Lambda}
\|\bar{\Phi}_\Lambda\|\,,
\end{align}
where $\|\Psi\|_X := \sum_{x \in X} \|\Psi_x\|$.
It is worth noting that the norm $\|\bar{\Phi}\|_{\varepsilon,\zeta}$ depends explicitly on $\Psi$.
Roughly speaking, the use of $\|\bar{\Phi}\|_{\varepsilon,\zeta}$ replaces the commutativity assumption between $\Psi$ and $\bar{\Phi}$ with a stronger control on the growth of $\|\bar{\Phi}_\Lambda\|$, which must now also account for the growth of $\|\Psi\|_\Lambda$.
With this strengthened norm, we are able to prove the following result.

\begin{theorem}\label{Thm: uniqueness for KMS quantum states}
	Let $\beta>0$ be such that
	\begin{align}\label{Eq: uniqueness for KMS quantum states - condition for uniqueness}
	\exists\varepsilon>0\mid
	\beta\|\bar{\Phi}\|_{\varepsilon+\log 3,2\beta}
	<\frac{1}{6}\frac{\varepsilon}{1+e^\varepsilon}\,,
	\end{align}
	
    Then there exists a unique $(\beta,\tau^\Gamma)$-KMS quantum state on $\mathcal{A}_\Gamma$.
\end{theorem}

The combination of Theorems \ref{Thm: uniqueness for KMS quantum states}-\ref{Thm: uniqueness for KMS classical states} establishes the existence of a common high-temperature regime in which uniqueness is granted for both KMS classical and quantum states.

\begin{corollary}\label{Cor: absence of classical and quantum phase transitions}
	Let assume that $\varphi$ fulfils \eqref{Eq: assumption on potential - C1 norm on potentials} and
	\begin{align}\label{Eq: condition for absence of classical and quantum phase transitions}
	\exists\varepsilon\mid
	\beta\|\bar{\varphi}\|_{\varepsilon+\log 3,2\beta}
	<\frac{1}{6}\frac{\varepsilon}{1+e^\varepsilon}\,,
	\end{align}
	where $\|\bar{\varphi}\|_{\varepsilon+\log 3,2\beta}$ is defined as \eqref{Eq: assumption on potential - norm on potentials strengthened}.
	Let $\Phi=Q_\Gamma(\varphi)$ be the associated family of quantized potentials.
	Then there exists a unique $(\beta,\delta_\Gamma)$-KMS classical state on $\mathscr{A}_\Gamma$ and a unique $(\beta,\tau^\Gamma)$-KMS quantum state on $\mathcal{A}_\Gamma$ for all $j\in\mathbb{Z}_+/2$.
\end{corollary}
Corollary \ref{Cor: absence of classical and quantum phase transitions} applies for all $\beta<\hat{\beta}_{\textsc{u}}$ where $\hat{\beta}_{\textsc{u}}$ is the unique value for such that
\[
\hat{\beta}_{\textsc{u}}\|\bar{\varphi}\|_{\varepsilon+\log 3,2\hat{\beta}_{\textsc{u}}}
=\frac{1}{6}\frac{\varepsilon}{1+e^\varepsilon}\,.
\]
The latter identifies a common non-trivial subcritical regime for both classical and quantum systems.
As explained in \cite{Drago_Pettinari_Van_de_Ven_2025}, Corollary \ref{Cor: absence of classical and quantum phase transitions} implies that, within the suitable framework of strict deformation quantization and considering the common subcritical regime, the semiclassical limit of the KMS quantum state coincides with the KMS classical state.

The paper is organized as follows.
Section \ref{Sec: Uniqueness result for quantum KMS states} is devoted to the statement and proof of Theorem \ref{Thm: uniqueness for KMS quantum states} (from which the proof of Theorem \ref{Thm: uniqueness for KMS quantum states - with commutation} also descends, \textit{cf.} Remark \ref{Rmk: proof of uniqueness for KMS quantum states - with commutation}).
This requires properly setting the stage for the quasi-local algebra of observables $\mathcal{A}_\Gamma$ and for the dynamics $\tau^\Gamma$ associated with the chosen family $\Phi$ of quantum potentials.
This construction slightly departs from the standard presentation of the topic, \textit{cf.} Section \ref{Subsec: dynamics for non-commuting Psi,Phi}.
In Section \ref{Subsec: decomposition in eta-free elements} we introduce a suitable decomposition of elements in $\mathcal{A}_\Gamma$, which plays a central role in the proof of Theorem \ref{Thm: uniqueness for KMS quantum states}.
The proof itself is presented in detail in Section \ref{Subsec: proof of uniqueness for quantum KMS states}.
Section \ref{Sec: consequences of (main) Theorem} discusses several consequences of Theorem \ref{Thm: uniqueness for KMS quantum states}.
In particular, in Section \ref{Subsec: Comparison of the subcritical inverse temperatures} we compare our subcritical inverse temperature $\beta_{\textsc{u}}$ with those obtained in \cite{Bratteli_Robinson_97}.
Section \ref{Subsec: common subcritical region for quantum and classical systems} is devoted to the proof of Theorem \ref{Thm: uniqueness for KMS classical states}, which serves as the classical analogue of Theorem \ref{Thm: uniqueness for KMS quantum states}.
This demonstrates that our proof is sufficiently robust to extend to the classical setting.
The paper concludes with the proof of Corollary \ref{Cor: absence of classical and quantum phase transitions}.

\paragraph{Acknowledgements.}
We thank D. Ueltschi for valuable discussions on this project and for his careful reading of an earlier version of this manuscript.
N.D. and L. P. are grateful for the support of the National Group of Mathematical Physics (GNFM-INdAM).
N.D. was supported by the MUR Excellence Department Project 2023-2027 awarded to the Department of Mathematics of the University of Genova, CUPD33C23001110001.

\paragraph{Data availability statement.}
Data sharing is not applicable to this article as no new data were created or analysed in this study.

\paragraph{Conflict of interest statement.}
The authors certify that they have no affiliations with or involvement in any
organization or entity with any financial interest or non-financial interest in
the subject matter discussed in this manuscript.


\section{Uniqueness result for quantum KMS states}
\label{Sec: Uniqueness result for quantum KMS states}

In this section we will state and prove the uniqueness result for KMS quantum states on quantum spin lattice system, \textit{cf.} Theorems \ref{Thm: uniqueness for KMS quantum states - with commutation}-\ref{Thm: uniqueness for KMS quantum states}.
In fact, we will only focus on the proof of Theorem \ref{Thm: uniqueness for KMS quantum states} because the one of Theorem \ref{Thm: uniqueness for KMS quantum states - with commutation} proceeds similarly.
Specifically, in section \ref{Subsec: quantum quasi-local algebra} we briefly recall the notion of quasi-local algebra $\mathcal{A}_\Gamma$ on an arbitrary countable set $\Gamma$ and the related notion of KMS quantum states while stating the precise content of Theorem \ref{Thm: uniqueness for KMS quantum states}.
Section \ref{Subsec: decomposition in eta-free elements} develops a purely algebraic decomposition of local observables in "$\eta$-free" elements, where $\eta=\{\eta_x\}_{x\in\Lambda}$ is a given family of states on the algebra $\mathcal{A}_x$ of local observables at site $x\in\Gamma$.
This decomposition is interesting per se, but it will turn rather helpful in the proof of Theorem \ref{Thm: uniqueness for KMS quantum states}, which is described in section \ref{Subsec: proof of uniqueness for quantum KMS states}.

\subsection{Quantum quasi-local algebra}
\label{Subsec: quantum quasi-local algebra}

In this section we will recall the setting for the description of the $C^*$-algebra of quasi-local observables over an arbitrary countable set $\Gamma$.
To begin with we consider, for any $x\in\Gamma$, a finite dimensional Hilbert space $\mathcal{H}_x$ and set, for $\Lambda\Subset\Gamma$,
\begin{align*}
	\mathcal{A}_\Lambda
	:=\mathcal{B}\Big(
	\bigotimes_{x\in\Lambda}\mathcal{H}_x\Big)\,.
\end{align*}
Elements of $\mathcal{A}_\Lambda$ are interpreted as observables localized within $\Lambda$.
There is a standard inclusion $\mathcal{A}_\Lambda\hookrightarrow\mathcal{A}_{\Lambda'}$, $\Lambda\subseteq\Lambda'\Subset\Gamma$, defined by
$A_\Lambda
\mapsto
A_\Lambda\bigotimes_{x\in\Lambda'\setminus\Lambda}1_x$.
With a common abuse of notation we will identify $\mathcal{A}_\Lambda$ with its image under this inclusion.
Out of the net of $C^*$-algebras $\{\mathcal{A}_\Lambda\}_{\Lambda\Subset\Gamma}$ we may consider the corresponding $C^*$-direct limit, which will be denoted by $\mathcal{A}_\Gamma$.
Elements of $\mathcal{A}_\Gamma$ are interpreted as quasi-local observables.
Since there will be no risk of confusion, we will denote by $\|\;\|$ the norms of either $\mathcal{A}_\Lambda$ or $\mathcal{A}_\Gamma$.
For later convenience we set $\dot{\mathcal{A}}_\Gamma:=\bigcup_{\Lambda\Subset\Gamma}\mathcal{A}_\Gamma$: The latter is a dense $*$-subalgebra of $\mathcal{A}_\Gamma$.

Next, one introduces a \textbf{dynamics} on $\mathcal{A}_\Gamma$, that is, a strongly continuous one-parameter group of automorphisms $\tau^\Gamma\colon\mathbb{R}\to\operatorname{Aut}(\mathcal{A}_\Gamma)$ on $\mathcal{A}_\Gamma$.
This is usually obtained by considering a family of potentials $\Phi=\{\Phi_\Lambda\}_{\Lambda\Subset\Gamma}$ so that, $\Phi_\varnothing=0$ and $\Phi_\Lambda=\Phi_\Lambda^*\in\mathcal{A}_\Lambda$ for all $\Lambda\Subset\Gamma$.
For a finite region $\Lambda\Subset\Gamma$ one defines a dynamics $\tau^\Lambda$ for the observables localized within $\Lambda$:
\begin{align*}
	\tau^\Lambda\colon\mathbb{R}\to\operatorname{Aut}(\mathcal{A}_\Lambda)\,,
	\qquad
	\tau^\Lambda_t(A_\Lambda)
	:=e^{itH_\Lambda}A_\Lambda e^{-itH_\Lambda}\,,
	\qquad
	H_\Lambda
	:=\sum_{X\subseteq\Lambda}\Phi_X\,.
\end{align*}
At this stage, one would like to obtain a dynamics on $\mathcal{A}_\Gamma$ by considering a suitable limit $\Lambda\uparrow\Gamma$ of the dynamics $\tau^\Lambda$ localized in finite regions.
This is possible under additional assumptions on $\Phi$.
Usually \cite{Bratteli_Robinson_97,Lenci_Rey_Bellet_2005}, it is assumed that
\begin{align} \label{Eq: assumption on potential - norm on potentials for dynamics}
	\exists\varepsilon>0\mid
	\|\Phi\|_\varepsilon<\infty\,,
\end{align}
where $\|\Phi\|_\varepsilon$ has been defined in \eqref{Eq: assumption on potential - norm on potentials}.
Specifically, out of assumption \eqref{Eq: assumption on potential - norm on potentials} one obtains the existence of a dynamics $\tau^\Gamma$ on $\mathcal{A}_\Gamma$, \textit{cf.} \cite[Thm. 6.2.4]{Bratteli_Robinson_97}.
The infinitesimal generator $\delta_\Gamma$ of $\tau^\Gamma$ reads
\begin{align}\label{Eq: dynamics on quasi-local algebra - infinitesimal generator}
\delta_\Gamma
\colon\dot{\mathcal{A}}_\Gamma
\to\mathcal{A}_\Gamma\,,
\qquad
\delta_\Gamma(A_\Lambda)
=i\sum_{\substack{X\Subset\Gamma\\X\cap\Lambda\neq\varnothing}}
[\Phi_X,A_\Lambda]
\quad\forall A_\Lambda\in\mathcal{A}_\Lambda
\subset\dot{\mathcal{A}}_\Gamma\,.
\end{align}
Moreover, one has
\[\tau^\Gamma_t
=\operatorname{s-}\lim_{\Lambda\uparrow\Gamma}
\tau^\Lambda_t\,.
\]
The fact that the dynamics $\tau^\Gamma$ is the strong limit of the family of dynamics $\tau^\Lambda$ has important consequences.
In particular, it ensures the existence of $\tau^\Gamma$-KMS states at any inverse temperature, \textit{cf.} \cite[Prop. 6.2.15]{Bratteli_Robinson_97}.

Assumption \eqref{Eq: assumption on potential - norm on potentials for dynamics} does not distinguish between contributions to the potential $\Phi$ coming from a single-site potential and those coming from multilocal potentials.
In fact, it is customary to split $\Phi$ in two families $\Psi$, $\bar{\Phi}$ given by
\begin{align}\label{Eq: decomposition of quantum potential}
\bar{\Phi}_\Lambda
:=\begin{dcases}
\Phi_\Lambda
&|\Lambda|\geq 2
\\
0
&|\Lambda|=1
\end{dcases}\,,
\qquad
\Psi_x:=\Phi_{\{x\}}\,.
\end{align}
This decomposition reflects the fact that the family $\bar{\Phi}:=\{\bar{\Phi}_\Lambda\}_{\Lambda\Subset\Gamma}$ collects all non-local potentials, whereas $\Psi=\{\Psi_x\}_{x\in\Gamma}$ contains only single-site potentials
\footnote{
	The results of this paper admit a slight generalization. Indeed, one may consider a decomposition where $\Psi=\{\Phi_{\Lambda_n}\}_n$ accounts for contributions from a family $\{\Lambda_n\}_n$ of disjoint finite regions covering $\Gamma$, that is, such that $\Lambda_n\cap\Lambda_m=\varnothing$ and $\cup_n\Lambda_n=\Gamma$.}.

At this stage, one may argue that the single-site potential $\Psi$ defines a dynamics $\tau^\Psi$ on $\mathcal{A}^\Gamma$ (regardless of condition \eqref{Eq: assumption on potential - norm on potentials for dynamics}) by setting
\begin{align}\label{Eq: dynamics on quasi-local algebra - single-site dynamics}
\tau^\Psi_t(A_\Lambda)
:=e^{it\Psi_\Lambda}A_\Lambda e^{-it\Psi_\Lambda}\,,
\qquad
\Psi_\Lambda
:=\sum_{x\in\Lambda}\Psi_x\,.
\end{align}
Notice that $\tau^\Psi_t$ leads to a well-defined dynamics on $\mathcal{A}^\Gamma$ because it preserves local observables ---that is, $\tau^\Psi_t(\mathcal{A}_\Lambda)\subset\mathcal{A}_\Lambda$--- and it commutes with the inclusion maps $\mathcal{A}_{\Lambda_0}\hookrightarrow\mathcal{A}_\Lambda$.
Roughly speaking, no thermodynamic limit is necessary because $\tau^\Psi$ preserves local observables.

It seems reasonable to seek for mild conditions on $\Psi$ and $\bar{\Phi}$ so that the dynamics $\tau^\Gamma$ is still well-defined regardeless the boundedness of $\|\Psi_x\|$ on $x\in\Gamma$ ---which is implied by \eqref{Eq: assumption on potential - norm on potentials for dynamics}.
A possible set of conditions is provided by
\begin{align}
\label{Eq: assumption on potential - commutation}
&[\bar{\Phi}_\Lambda,\Psi_x]=0
\qquad
\forall x\in\Gamma\,,\,
\forall\Lambda\Subset\Gamma\,,
\\
	\label{Eq: assumption on potential - norm on multilocal potentials}
&\exists\varepsilon>0\mid
\|\bar{\Phi}\|_\varepsilon<\infty\,.
\end{align}
Indeed, condition \eqref{Eq: assumption on potential - norm on multilocal potentials} ensures, again by \cite[Thm. 6.2.4]{Bratteli_Robinson_97}, the existence of a dynamics $\bar{\tau}^\Gamma$ with infinitesimal generator $\bar{\delta}_\Gamma$ given by \eqref{Eq: dynamics on quasi-local algebra - infinitesimal generator} (with $\Phi$ replaced by $\bar{\Phi}$) and such that
$\bar{\tau}^\Gamma_t
=\operatorname{s-}\lim_{\Lambda\uparrow\Gamma}
\bar{\tau}^\Lambda_t$.
Moreover, assumption \eqref{Eq: assumption on potential - commutation} guarantees that $\tau^\Lambda_t=\bar{\tau}^\Lambda_t\tau^\Psi_t$, therefore, the dynamics $\tau^\Gamma$ can be obtained by
\begin{align}\label{Eq: dynamics on quasi-local algebra}
\tau^\Gamma_t
:=\bar{\tau}^\Gamma_t
\circ\tau^\Psi_t
=\operatorname{s-}\lim_{\Lambda\uparrow\Gamma}
\bar{\tau}^\Lambda_t
\circ\tau^\Psi_t
=\operatorname{s-}\lim_{\Lambda\uparrow\Gamma}
\tau^\Lambda_t\,.
\end{align}
with infinitesimal generator \eqref{Eq: dynamics on quasi-local algebra - infinitesimal generator}.

\subsection{Dynamics for non-commuting $\Psi,\bar{\Phi}$}
\label{Subsec: dynamics for non-commuting Psi,Phi}

Assumptions \eqref{Eq: assumption on potential - norm on potentials for dynamics} 
and \eqref{Eq: assumption on potential - commutation}-\eqref{Eq: assumption on potential - norm on potentials} 
are well known and widely used in the literature \cite{Bratteli_Robinson_97,Lenci_Rey_Bellet_2005}.
As discussed in the introduction, however, the present work adopts a slightly different framework.
In particular, we aim to dispense with assumption \eqref{Eq: assumption on potential - commutation}.

This choice is motivated by the results presented in 
Sections \ref{Subsec: Uniqueness result for classical KMS states}--\ref{Subsec: common subcritical region for quantum and classical systems}, 
where we establish a classical analogue of Theorem \ref{Thm: uniqueness for KMS quantum states} 
(\textit{cf.} Theorem \ref{Thm: uniqueness for KMS classical states}), 
as well as a common high-temperature regime ensuring uniqueness of both classical and quantum KMS states 
(\textit{cf.} Corollary \ref{Cor: absence of classical and quantum phase transitions}).
To obtain the latter result, it is necessary to compare the assumptions in 
Theorems \ref{Thm: uniqueness for KMS quantum states} 
and \ref{Thm: uniqueness for KMS classical states}.
This comparison shows that assumption \eqref{Eq: assumption on potential - commutation} 
is difficult to accommodate within the framework of strict deformation quantization, \textit{cf.} Section \ref{Subsec: common subcritical region for quantum and classical systems}.

Dropping assumption \eqref{Eq: assumption on potential - commutation} still allows one to define a meaningful dynamics $\tau^\Gamma$ on $\mathcal{A}_\Gamma$.  
Indeed, it suffices to define $\tau^\Gamma$ through its associated Dyson series.
For $A_\Lambda\in\mathcal{A}_\Lambda$, $\Lambda\Subset\Gamma$, we set
\begin{align}\label{Eq: dynamics on quasi-local algebra - Dyson series}
	\tau^\Gamma_t(A_\Lambda)
	:=\sum_{n\geq 0}i^n
	\int_0^t\mathrm{d}t_1
	\cdots
	\int_0^{t_{n-1}}\mathrm{d}t_n
	\sum_{X_1,\ldots,X_n:\Lambda}
	\operatorname{ad}_{\tau^\Psi_{t_1}\bar{\Phi}_{X_1}}
	\cdots
	\operatorname{ad}_{\tau^\Psi_{t_n}\bar{\Phi}_{X_n}}
	\tau^\Psi_t(A_\Lambda)\,,
\end{align}
where $\operatorname{ad}_A(B):=[A,B]$ while $\sum_{X_1,\ldots,X_n:\Lambda}$ denotes the sum over $X_1,\ldots,X_n\Subset\Gamma$ subjected to the constraint $X_j\cap S_{j-1}^\Lambda\neq\varnothing$ for all $j\in\{1,\ldots,n\}$, where
\begin{align}\label{Eq: hyerarchy of sets}
	S_0^\Lambda:=\Lambda\,,
	\qquad
	S_j^\Lambda:=X_j\cup S_{j-1}^\Lambda\,.
\end{align}
Convergence of the series above is ensured by assumption \eqref{Eq: assumption on potential - norm on potentials}, together with the standard estimate stated in the following lemma, which we include here for later reference.

\begin{lemma}\label{Lem: sum on subsets estimate}
	Let $(\alpha_X)_{X\Subset\Gamma}$ be a sequence of positive numbers and $\Lambda\Subset\Gamma$.
	Then for all $\varepsilon>0$ we have
	\begin{align}\label{Eq: sum on subsets estimate}
		\sum_{X_1,\ldots,X_n:\Lambda}
		\prod_{j=1}^n\alpha_{X_j}
		\leq
		n!\varepsilon^{-n}e^{\varepsilon|\Lambda|}
		\Big(\sup_{x\in\Gamma}\sum_{\substack{X\Subset\Gamma\\x\in X}}e^{\varepsilon(|X|-1)}\alpha_{X}\Big)^n\,.
	\end{align}
\end{lemma}
\begin{proof}
	To begin with one observe that
	\begin{align*}
		\sum_{X\cap S\neq\varnothing}\alpha_X
		\leq|S|\sup_{x\in\Gamma}\sum_{X\ni x}\alpha_X\,,
	\end{align*}
	which immediately leads to
	\begin{multline*}
		\sum_{X_1,\ldots,X_n:\Lambda}
		\prod_{j=1}^n\alpha_{X_j}
		=\sum_{\substack{X_1\Subset\Gamma\\X_1\cap S_0^\Lambda\neq\varnothing}}
		\alpha_{X_1}
		\sum_{\substack{X_2\Subset\Gamma\\X_2\cap S_1^\Lambda\neq\varnothing}}
		\alpha_{X_2}
		\ldots
		\sum_{\substack{X_n\Subset\Gamma\\X_n\cap S_{n-1}^\Lambda\neq\varnothing}}
		\alpha_{X_n}
		\\
		\leq\sup_{x_1\in\Gamma}\sum_{\substack{X_1\Subset\Gamma\\x_1\in X_1}}\alpha_{X_1}
		\ldots
		\sup_{x_n\in\Gamma}\sum_{\substack{X_n\Subset\Gamma\\x_n\in X_n}}\alpha_{X_n}
		\prod_{j=0}^{n-1}|S_j^\Lambda|\,.
	\end{multline*}
	Moreover, by induction we have $|S_j^\Lambda|\leq(|\Lambda|+|X_1|+\ldots+|X_j|-j)$.
    Using the elementary bound $x^n \leq n!\,\varepsilon^{-n} e^{\varepsilon x}$, $x\geq 0$, we therefore obtain
	\begin{align*}
		\prod_{j=0}^{n-1}|S_j^\Lambda|
		\leq(|\Lambda|+|X_1|+\ldots+|X_n|-n)^n
		\leq n!\varepsilon^{-n}e^{\varepsilon|\Lambda|}
		\prod_{j=1}^ne^{\varepsilon(|X_j|-1)}\,.
	\end{align*}
	Inserting this estimate in the above inequality proves \eqref{Eq: sum on subsets estimate}.
\end{proof}
Lemma \ref{Lem: sum on subsets estimate} implies that
\begin{align*}
	\|\tau^\Gamma_t(A_\Lambda)\|
	\leq\|A_\Lambda\|\sum_{n\geq 0}
	\frac{(2|t|)^n}{n!}
	\sum_{X_1,\ldots,X_n:\Lambda}
	\prod_{j=1}^n\|\bar{\Phi}_{X_j}\|
	\leq\|A_\Lambda\|e^{\varepsilon|\Lambda|}\sum_{n\geq 0}
	\Big(\frac{2|t|}{\varepsilon}\|\bar{\Phi}\|_\varepsilon\Big)^n\,,
\end{align*}
proving convergence of \eqref{Eq: dynamics on quasi-local algebra - Dyson series} for $2|t|\|\bar{\Phi}\|_\varepsilon<\varepsilon$ ---we are still assuming \eqref{Eq: assumption on potential - norm on potentials}.
With a similar estimate one proves that $t\mapsto\tau^\Gamma_t(A_\Lambda)$ is strongly continuous, moreover, that it fulfills $\frac{\mathrm{d}}{\mathrm{d}t}\tau^\Gamma_t(A_\Lambda)=i\delta_\Gamma(\tau^\Gamma(A_\Lambda))$.
This implies that $\tau^\Gamma$ extends to a one-parameter group of automorphisms.

The drawback of this construction is that we lack control over the thermodynamic limit.  
In particular, it is not clear \textit{a priori} whether $\tau^\Gamma$ coincides with the strong limit of the local dynamics $\tau^\Lambda$ as $\Lambda\uparrow\Gamma$.
This is problematic, since it prevents us from guaranteeing the existence of $\tau^\Gamma$-KMS states.
Indeed, if 
\begin{align*}
	\tau^\Gamma_t=\operatorname{s-}\lim_{\Lambda\uparrow\Gamma}\tau^\Lambda_t\,,
\end{align*}
then every weak$^*$-limit point of a sequence $(\omega_\Lambda^\beta)_\Lambda$ of $(\beta,\tau^\Lambda)$-KMS states is automatically a $(\beta,\tau^\Gamma)$-KMS state, \textit{cf.} \cite[Prop. 6.2.15]{Bratteli_Robinson_97}.
Existence of $\tau^\Lambda$-KMS states is not an issue because for every finite $\Lambda\Subset\Gamma$ there exists a unique $(\beta,\tau^\Lambda)$-KMS state, namely the Gibbs state $\omega_\Lambda^\beta(A_\Lambda)=\operatorname{tr}(e^{-\beta H_\Lambda}A_\Lambda)/\operatorname{tr}(e^{-\beta H_\Lambda})$.
The existence of weak$^*$-limit points for the sequence $(\omega_\Lambda^\beta)_\Lambda$ follows from a weak$^*$-compactness argument.

To ensure that $\tau^\Gamma$ is indeed the strong limit of $\tau^\Lambda$ as $\Lambda\uparrow\Gamma$, we need to strengthen assumption \eqref{Eq: assumption on potential - norm on potentials}.
Specifically we will assume that
\begin{align}\label{Eq: assumption on potential - norm on potentials strengthened for dynamics}
	\exists\varepsilon,\zeta>0\mid
	\|\bar{\Phi}\|_{\varepsilon,\zeta}
	<+\infty\,,
\end{align}
where $\|\bar{\Phi}\|_{\varepsilon,\zeta}$ has been defined in \eqref{Eq: assumption on potential - norm on potentials strengthened}.
As already mentioned in the introduction, it is worth stressing that $\|\bar{\Phi}\|_{\varepsilon,\zeta}$ implicitly depends on $\Psi$.
Finiteness of $\|\bar{\Phi}\|_{\varepsilon,\zeta}$ ensures that the possible growth of $\|\Psi\|_X$ is always compensated by $\|\bar{\Phi}_X\|$.
This strong requirement compensates for the lack of commutativity between $\Psi$ and $\bar{\Phi}$.

Notice that, if $\|\Psi\|:=\sup_{x\in\Gamma}\|\Psi_x\|<\infty$, condition \eqref{Eq: assumption on potential - norm on potentials strengthened for dynamics} essentially reduce to \eqref{Eq: assumption on potential - norm on potentials for dynamics} because of the estimate
\[
	\|\bar{\Phi}\|_{\varepsilon,\zeta}
	\leq e^{\zeta\|\Psi\|}\|\bar{\Phi}\|_{\varepsilon+\zeta\|\Psi\|}\,.
\]
In the general case, however, condition \eqref{Eq: assumption on potential - norm on potentials strengthened for dynamics} also covers situations in which $\|\Psi_x\|$ is unbounded and $\Psi$ does not commute with $\bar{\Phi}$, as the following example clearly illustrates.

\begin{example}
	Consider the two-dimensional quantum Ising model on $\Gamma=\mathbb{Z}^2$ with an unbounded transverse magnetic potential along the line $x=(0,x_2)$:
	\begin{equation}\label{eq: unbounded potential}
		\Phi_\Lambda :=
		\begin{dcases}
			t_{x,y} S^3_x S^3_y & \Lambda = \{x,y\}\in \Gamma_b,\\
			B|x_2|S^1_x & \Lambda = \{(0,x_2)\},\\
			0 & \text{otherwise}.
		\end{dcases}
	\end{equation}
	Here, $\Gamma_b$ denotes the set of pairs of points $x,y\in\mathbb{Z}^2$ with $\mathrm{d}(x,y)=1$, $\mathrm{d}$ is the graph distance, and $B\neq 0$ and $t_{x,y}$ are real parameters.
	Notice that $\Phi$ does not fulfil condition \eqref{Eq: assumption on potential - commutation}.
	Assume that there exists $a>0$ such that
	\begin{equation*}
		\sup_{(x,y)\in\Gamma_b}|t_{x,y}|<+\infty,
		\qquad
		\limsup_{|x_2|\to +\infty}\sup_{y\in\mathbb{Z}^2}|t_{(0,x_2),y}|e^{a|x_2|}<+\infty.
	\end{equation*}
	The second condition means that the interaction strength decays along the line on which the magnetic potential becomes large.
	A straightforward computation then shows that
	\begin{equation*}
		\|\bar{\Phi}\|_{\varepsilon,a/2|B|}<+\infty\,,
	\end{equation*}
	for every $\varepsilon>0$.
	Therefore, one can construct a dynamics $\tau^\Gamma$ implementing the potential in \eqref{eq: unbounded potential} as the strong limit of the finite-volume dynamics $\tau^\Lambda$.
\end{example}

To establish that the dynamics $\tau^\Gamma$ arises as the strong limit of the finite-volume dynamics $\tau^\Lambda$ as $\Lambda\uparrow\Gamma$, we invoke \cite[Prop. 6.2.3]{Bratteli_Robinson_97}.  
This result guarantees that the (closure of the) derivation $\delta_\Gamma$ generates a dynamics $\tau^\Gamma$ on $\mathcal{A}_\Gamma$, with $\tau^\Gamma_t = \operatorname{s-}\lim_{\Lambda\uparrow\Gamma}\tau^\Lambda_t$, where the convergence is uniform for $t$ in compact subsets of $\mathbb{R}$, provided that $\delta_\Gamma$ admits a dense set of analytic elements in $\mathcal{A}_\Gamma$.  
The following result establishes precisely this property, with a proof obtained by adapting the arguments of \cite[Thm. 6.2.4]{Bratteli_Robinson_97} to the present setting.

\begin{proposition}\label{Prop: dynamics on quasi-local algebra - existence with strengthened assumptions}
	If assumption \eqref{Eq: assumption on potential - norm on potentials strengthened for dynamics} holds, then $\dot{\mathcal{A}}_\Gamma$ is a dense $*$-subalgebra of analytic elements for the closure of $\delta_\Gamma$.
\end{proposition}

\begin{proof}
	To begin with, we observe that $\delta_\Gamma=\delta_\Psi+\bar{\delta}_\Gamma$, where $\delta_\Psi(A_\Lambda):=i[\Psi_\Lambda,A_\Lambda]$ while $\bar{\delta}_\Gamma$ is defined as in \eqref{Eq: dynamics on quasi-local algebra - infinitesimal generator} with $\Phi$ replaced by $\bar{\Phi}$.
	Well-definiteness of $\bar{\delta}_\Gamma$ follows from assumption \eqref{Eq: assumption on potential - norm on potentials strengthened}, in particular the series defining $\bar{\delta}_\Gamma(A_\Lambda)$ converges on account of the estimate
	\begin{align*}
		\|\bar{\delta}_\Gamma(A_\Lambda)\|
		\leq 2\|A_\Lambda\||\Lambda
		|\sup_{x\in\Gamma}\sum_{\substack{X\Subset\Gamma\\x\in X}}\|\Phi_X\|
		\leq 2\|A_\Lambda\||\Lambda|
		\|\bar{\Phi}\|_{\varepsilon,\zeta}\,.
	\end{align*}
	We will now prove that, for all $A_\Lambda\in\mathcal{A}_\Lambda$, the series
	\begin{align*}
		\sum_{n\geq 0}\frac{z^n}{n!}\|\delta_\Gamma^n(A_\Lambda)\|\,,
	\end{align*}
	has a non-vanishing radius of convergence.
	This will be proved by deriving a suitable estimate for $\|\delta_\Gamma^n(A_\Lambda)\|$.
	To begin with we observe that
	\begin{align*}
		\delta_\Gamma^n(A_\Lambda)
		=\sum_{k=0}^n\sum_{\substack{m_0,\ldots,m_k\in\mathbb{Z}_+\\m_0+\ldots+m_k=n-k}}
		\delta_\Psi^{m_0}\bar{\delta}_\Gamma
		\delta_\Psi^{m_1}
		\cdots\bar{\delta}_\Gamma
		\delta_\Psi^{m_k}(A_\Lambda)\,.
	\end{align*}
	Out of the explicit definition of $\delta_\Psi$ and $\bar{\delta}_\Gamma$ we have
	\begin{multline*}
		\|\delta_\Gamma^n(A_\Lambda)\|
		\leq\sum_{k=0}^n\sum_{\substack{m_0,\ldots,m_k\in\mathbb{Z}_+\\m_0+\ldots+m_k=n-k}}
		\\
		\sum_{X_1,\ldots,X_k:\Lambda}
		\|(\operatorname{ad}_{\Psi_{S_k^\Lambda}})^{m_0}
		\operatorname{ad}_{\bar{\Phi}_{X_k}}
		(\operatorname{ad}_{\Psi_{S_{k-1}^\Lambda}})^{m_1}
		\cdots\operatorname{ad}_{\bar{\Phi}_{X_1}}
		(\operatorname{ad}_{\Psi_{S^\Lambda_0}})^{m_k}(A_\Lambda)\|\,.
	\end{multline*}
	Since $\|\Psi_{S_j^\Lambda}\|\leq\|\Psi\|_{S_k^\Lambda}$ for all $j\in\{1,\ldots,k\}$ we have
	\begin{align*}
		\|(\operatorname{ad}_{\Psi_{S_k^\Lambda}})^{m_0}
		\operatorname{ad}_{\bar{\Phi}_{X_k}}
		(\operatorname{ad}_{\Psi_{S_{k-1}^\Lambda}})^{m_1}
		\cdots\operatorname{ad}_{\bar{\Phi}_{X_1}}
		(\operatorname{ad}_{\Psi_{S^\Lambda_0}})^{m_k}(A_\Lambda)\|
		\leq\|A_\Lambda\|
		2^n\|\Psi\|_{S_k^\Lambda}^{n-k}
		\prod_{j=1}^k\|\bar{\Phi}_{X_j}\|\,.
	\end{align*}
	Furthermore, for all $\zeta>0$ it holds
    \begin{align*}
		\|\Psi\|_{S^\Lambda_k}^{n-k}
		\leq(\|\Psi\|_\Lambda
		+\|\Psi\|_{X_1}
		+\ldots
		+\|\Psi\|_{X_k})^{n-k}
		\leq (n-k)!\zeta^{-(n-k)}
		e^{\zeta\|\Psi\|_\Lambda}\prod_{j=1}^ke^{\zeta\|\Psi\|_{X_j}}\,,
	\end{align*}
	which implies
	\begin{align*}
		\|\delta_\Gamma^n(A_\Lambda)\|
		\leq\|A_\Lambda\|e^{\zeta\|\Psi\|_\Lambda}
		n!2^n\zeta^{-n}
		\sum_{k=0}^n
		\frac{\zeta^k}{k!}
		\sum_{X_1,\ldots,X_k:\Lambda}
		\prod_{j=1}^k
		e^{\zeta\|\Psi\|_{X_j}}
		\|\bar{\Phi}_{X_j}\|\,,
	\end{align*}
	where the sum over $m_0,\ldots,m_k$ has provided a factor ${n\choose k}$.
	An application of Lemma \ref{Lem: sum on subsets estimate} leads to
	\begin{align*}
		\|\delta_\Gamma^n(A_\Lambda)\|
		&\leq
		\|A_\Lambda\|e^{\zeta\|\Psi\|_\Lambda}
		e^{\varepsilon|\Lambda|}
		n!2^n\zeta^{-n}
		\sum_{k=0}^n
		\Big(\frac{\zeta}{\varepsilon}
		\|\bar{\Phi}\|_{\varepsilon,\zeta}\Big)^k
		\\
		&=
		\|A_\Lambda\|e^{\zeta\|\Psi\|_\Lambda}
		e^{\varepsilon|\Lambda|}
		n!2^n\zeta^{-n}
		\frac{1-\Big(\frac{\zeta}{\varepsilon}\|\bar{\Phi}\|_{\varepsilon,\zeta}\Big)^{n+1}}
		{1-\frac{\zeta}{\varepsilon}\|\bar{\Phi}\|_{\varepsilon,\zeta}}\,.
	\end{align*}
	The latter estimate implies
	\begin{align*}
		\sum_{n\geq 0}\frac{|z|^n}{n!}\|\delta_\Gamma^n(A_\Lambda)\|
		\leq
		\frac{\|A_\Lambda\|e^{\zeta\|\Psi\|_\Lambda}e^{\varepsilon|\Lambda|}}{1-\frac{\zeta}{\varepsilon}\|\bar{\Phi}\|_{\varepsilon,\zeta}}
		\sum_{n\geq 0}\Big(\frac{2|z|}{\zeta}\Big)^n
		\Big[1-\Big(\frac{\zeta}{\varepsilon}\|\bar{\Phi}\|_{\varepsilon,\zeta}\Big)^{n+1}\Big]
		<+\infty\,,
	\end{align*}
	for $|z|<\min\{\zeta/2,\varepsilon/2\|\bar{\Phi}\|_{\varepsilon,\zeta}\}$.
\end{proof}

\begin{remark}\label{Rmk: Dyson series - analytic continuation}
	\noindent
	\begin{enumerate}[(i)]
		\item
		The existence of the thermodynamic limit 
		$\tau^\Gamma_t=\operatorname{s-}\lim_{\Lambda\uparrow\Gamma}\tau^\Lambda_t$ 
		can be established by means of Lieb-Robinson bounds \cite{Hastings_Koma_2006,Lieb_Robinson_1972,Nachtergaele_Ogata_Sims_2006,Nachtergaele_Raz_Schlein_Sims_2009,Nachtergaele_Sims_2006,Nachtergaele_Sims_2014,Nachtergaele_Schlein_Sims_Starr_Zagrebnov_2010}.
		In that framework, the assumptions on $\|\bar{\Phi}_X\|$ are more directly related to the diameter of $X$ rather than its cardinality, 
		leading to conditions that are more suitable for physical applications. 
		However, to the best of our knowledge, the assumptions employed in the derivation of Lieb-Robinson bounds do not suffice to prove the uniqueness of KMS quantum states at sufficiently high temperatures. 
		Since this is ultimately our goal, we therefore prefer to work under assumption \eqref{Eq: assumption on potential - norm on potentials strengthened for dynamics}.

		\item 
		Assumption \eqref{Eq: assumption on potential - norm on potentials strengthened for dynamics} implies that the Dyson series \eqref{Eq: dynamics on quasi-local algebra - Dyson series} can be extended to the complex plane and specifically in the complex imaginary direction $t=i\sigma$ for small enough $\sigma>0$.
		Indeed, since $\|\tau^\Psi_{i\sigma}(A_\Lambda)\|\leq e^{2\sigma\|\Psi\|_\Lambda}\|A_\Lambda\|$, we have
		\begin{align*}
			\|\tau^\Gamma_{i\sigma}(A_\Lambda)\|
			&\leq\sum_{n\geq 0}
			\int_0^\sigma\mathrm{d}s_1
			\cdots
			\int_0^{s_{n-1}}\mathrm{d}s_n
			\sum_{X_1,\ldots,X_n:\Lambda}
			\|\operatorname{ad}_{\tau^\Psi_{is_1}\bar{\Phi}_{X_1}}
			\cdots
			\operatorname{ad}_{\tau^\Psi_{is_n}\bar{\Phi}_{X_n}}
			\tau^\Psi_{i\sigma}(A_\Lambda)\|
			\\
			&\leq\|A_\Lambda\|e^{2\sigma\|\Psi\|_\Lambda}
			\sum_{n\geq 0}
			\frac{(2\sigma)^n}{n!}
			\sum_{X_1,\ldots,X_n:\Lambda}
			\prod_{j=1}^ne^{2\sigma\|\Psi\|_{X_j}}\|\bar{\Phi}_{X_j}\|
			\\
			&\leq\|A_\Lambda\|e^{2\sigma\|\Psi\|_\Lambda}
			e^{\varepsilon|\Lambda|}
			\sum_{n\geq 0}
			\Big(\frac{2\sigma}{\varepsilon}\|\bar{\Phi}\|_{\varepsilon,2\sigma}\Big)^n\,,
		\end{align*}
		which converges for for any $2\sigma<\zeta$ such that $2\sigma\|\bar{\Phi}\|_{\varepsilon,2\sigma}<\varepsilon$.
	\end{enumerate}
\end{remark}

In what follows, we will primarily focus on KMS quantum states on $\mathcal{A}_\Gamma$ associated with the dynamics $\tau^\Gamma$.  
For a fixed inverse temperature $\beta>0$, we refer to these as $(\beta,\tau^\Gamma)$-KMS quantum states.  
The existence of such states follows from a standard weak$^*$-compactness argument: any weak$^*$-limit point of $\tau^\Lambda$-KMS states on $\mathcal{A}_\Lambda$ yields a $\tau^\Gamma$-KMS state on $\mathcal{A}_\Gamma$, \textit{cf.} \cite[Prop. 6.2.15]{Bratteli_Robinson_97}. 
Our goal is to prove Theorem \ref{Thm: uniqueness for KMS quantum states} which implies uniqueness of $(\beta,\tau^\Gamma)$-KMS quantum states for small enough $\beta$, \textit{cf.} condition \eqref{Eq: uniqueness for KMS quantum states - condition for uniqueness}.

\begin{remark}\label{Rmk: subcritical inverse temperature}
    Notice that the function $\beta \mapsto \beta \|\bar{\Phi}\|_{\varepsilon+\log 3,2\beta}$ is strictly increasing and diverges as $\beta \to \infty$. 
    Therefore, if condition \eqref{Eq: uniqueness for KMS quantum states - condition for uniqueness} holds for some $\beta_0$, then it is also satisfied for all $\beta < \beta_0$.
    From this point of view, Theorem \ref{Thm: uniqueness for KMS quantum states} identifies an "optimal" subcritical inverse temperature $\beta_{\textsc{u}}$, defined as the unique value such that
    \begin{align}\label{Eq: estimate of subcritical inverse temperature}
        \beta_{\textsc{u}} \|\bar{\Phi}\|_{\varepsilon+\log 3,2\beta_{\textsc{u}}}
        = \frac{1}{6} \frac{\varepsilon}{1+e^\varepsilon}\,.
    \end{align}
    Indeed, uniqueness of $(\beta,\tau^\Gamma)$-KMS states is guaranteed for all $\beta < \beta_{\textsc{u}}$.
\end{remark}

\subsection{Decomposition in $\eta$-free elements}
\label{Subsec: decomposition in eta-free elements}

The proof of Theorem \ref{Thm: uniqueness for KMS quantum states} makes use of the following technical construction, which has been inspired by \cite{{Drago_Van_de_Ven_2024}} and it is interesting in its own right.
For each $x\in\Gamma$, let $\eta_x\colon\mathcal{A}_x\to\mathbb{C}$, be a state on $\mathcal{A}_x$.
We will denote by $\eta_X:=\bigotimes_{x\in X}\eta_x\colon\mathcal{A}_X\to\mathbb{C}$.
For a smaller set $X\subseteq\Lambda$, we extend $\eta_X$ to $\mathcal{A}_\Lambda$ by considering $I_{\Lambda\setminus X} \otimes \eta_X\colon\mathcal{A}_\Lambda\to\mathcal{A}_{\Lambda\setminus X}$.
For simplicity, we reuse the symbol $\eta_X$ for this extended map.
We then define
\begin{align}\label{Eq: eta-free elements - definition}
	\tilde{\mathcal{A}}_\Lambda
	:=\{\tilde{A}_\Lambda\in\mathcal{A}_\Lambda\mid
	\eta_x(\tilde{A}_\Lambda)=0\;\forall x\in \Lambda\}\,,
\end{align}
while, by definition, $\tilde{\mathcal{A}}_\varnothing$ consists of multiples of the unit element of $\mathcal{A}_\Gamma$.
Elements in $\mathcal{A}_\Lambda$ will be called \textbf{$\eta$-free within $\Lambda$}.

\begin{remark}\label{Rmk: decomposition in eta-free elements - generic case}
	The results of this section, \textit{cf.} Proposition \ref{Prop: decomposition in eta-free elements}, can be generalized to the case where $\eta_x\colon\mathcal{A}_x\to\mathcal{B}_x$ is a conditional expectation on a given $C^*$-subalgebra $\mathcal{B}_x$ of $\mathcal{A}_x$.
    We recall that the latter is a linear, positive and normalized map such that $\eta_x(B_1AB_2)=B_1\eta_x(A)B_2$ for all $A\in\mathcal{A}_x$ and $B_1,B_2\in\mathcal{B}_x$.
    Within this setting, the proper generalization of $\tilde{\mathcal{A}}_X$ is as follows:
    \begin{align*}
    	\mathcal{A}_\Lambda^X
    	:=\{A_\Lambda^X\in\mathcal{A}_\Lambda\mid
    	\eta_x(A_\Lambda^X)=0\;\forall x\in X\,,\,
    	\eta_x(A_\Lambda^X)=A_\Lambda^X\;\forall x\notin X\}\,,
    \end{align*}
    while, by definition, $\mathcal{A}_\varnothing^\varnothing$ consists of multiples of the unit element of $\mathcal{A}_\Gamma$.
    Mind the different meanings of $X$ and $\Lambda$ in $\mathcal{A}_\Lambda^X$: $\Lambda$ denotes the localization of $\mathcal{A}_\Lambda^X$, 
    while $X$ refers to the region where $\mathcal{A}_\Lambda^X$ is $\eta$-free.
    The second condition in the definition of $\mathcal{A}_\Lambda^X$ is necessary to ensure that $\mathcal{A}_\Lambda^X \cap \mathcal{A}_\Lambda^Y = \varnothing$ for $X\neq Y$.
\end{remark}

Crucially, any element $A_\Lambda \in \mathcal{A}_\Lambda$ can be decomposed as a finite sum of elements in $\tilde{\mathcal{A}}_X$ for various $X \subseteq \Lambda$.

\begin{proposition}\label{Prop: decomposition in eta-free elements}
	Let $\Lambda\Subset\Gamma$ and $A_\Lambda\in \mathcal{A}_\Lambda$.
	Then there exists a unique sequence $(\tilde{A}_X)_{X\subseteq\Lambda}$ with $\tilde{A}_X\in\tilde{\mathcal{A}}_X$ such that
	\begin{align}\label{Eq: decomposition in eta-free elements}
		A_\Lambda
		=\sum_{X\subseteq\Lambda}\tilde{A}_X\,.
	\end{align}
	The elements $\tilde{A}_X\in\tilde{\mathcal{A}}_X$ are recursively defined by
	\begin{align}\label{Eq: decomposition in eta-free elements - recursive definition}
		\tilde{A}_\varnothing
		:=\eta_\Lambda(A_\Lambda)\,,
		\qquad
		\tilde{A}_X
		:=\eta_{\Lambda\setminus X}(A_\Lambda)
		-\sum_{Y\subset X}\tilde{A}_Y
		\qquad
		\forall X\subseteq\Lambda\,,
	\end{align}
	where we set $\eta_\varnothing(A_\Lambda):=A_\Lambda$.
	This recursive relation can be inverted and in particular we have
	\begin{align}\label{Eq: decomposition in eta-free elements - explicit form}
		\tilde{A}_X
		=\sum_{Y\subseteq X}(-1)^{|X\setminus Y|}\eta_{\Lambda\setminus Y}(A_\Lambda)\,.
	\end{align}
\end{proposition}
\begin{proof}
	Equation \eqref{Eq: decomposition in eta-free elements} follows from equation \eqref{Eq: decomposition in eta-free elements - recursive definition} for $X=\Lambda$.
	We will now prove that $\tilde{A}_X$ belongs to $\tilde{\mathcal{A}}_X$ for all $X\subseteq\Lambda$.
	This can be proved by induction over $|X|$.
	If $|X|=0$ ---that is, if $X=\varnothing$--- we find $\tilde{A}_\varnothing:=\eta_\Lambda(A_\Lambda)\in\tilde{\mathcal{A}}_\varnothing$.
	Let now assume that $\tilde{A}_X\in\tilde{\mathcal{A}}_X$ for all $X\subseteq\Lambda$ with $|X|\leq N-1$.
	Let $X\subseteq\Lambda$ with $|X|=N$.
	We will prove that $\tilde{A}_X\in\tilde{\mathcal{A}}_X$.
	Let $x\in X$: Then equation \eqref{Eq: decomposition in eta-free elements - recursive definition} implies
	\begin{align*}
		\eta_x(\tilde{A}_X)
		=\eta_{\Lambda\setminus(X\setminus\{x\})}(A_\Lambda)
		-\sum_{Y\subset X}\eta_x(\tilde{A}_Y)\,.
	\end{align*}
	Since $|Y|<N$ we have $\tilde{A}_Y\in\tilde{\mathcal{A}}_Y$.
	This implies $\eta_x(\tilde{A}_Y)=0$ for all $x\in Y$.
	It follows that
	\begin{align*}
		\eta_x(\tilde{A}_X)
		=\eta_{\Lambda\setminus(X\setminus\{x\})}(A_\Lambda)
		-\sum_{Y\subseteq X\setminus\{x\}}\tilde{A}_Y
		=\eta_{\Lambda\setminus(X\setminus\{x\})}(A_\Lambda)
		-\sum_{Y\subset X\setminus\{x\}}\tilde{A}_Y
		-\tilde{A}_{X\setminus\{x\}}
		=0\,,
	\end{align*}
	where we again used equation \eqref{Eq: decomposition in eta-free elements - recursive definition} for $\tilde{A}_{X\setminus\{x\}}$.

    It remains to prove equation \eqref{Eq: decomposition in eta-free elements - explicit form}.
	We proceed again by induction on $|X|$.
	For $X=\varnothing$ we have $\tilde{A}_\varnothing=\eta_\Lambda(A_\Lambda)$, which is consistent with equation \eqref{Eq: decomposition in eta-free elements - explicit form}.
	By induction, let assume that equation \eqref{Eq: decomposition in eta-free elements - explicit form} has been proved for all $X\subseteq\Lambda$ with $|X|=N-1$.
	Let $X\subseteq\Lambda$ with $|X|=N$.
	We will prove equation \eqref{Eq: decomposition in eta-free elements - explicit form} for such $X$.
	We have
	\begin{align*}
		\tilde{A}_X
		&=\eta_{\Lambda\setminus X}(A_\Lambda)
		-\sum_{Y\subset X}\tilde{A}_Y
		&\textrm{eq. }\eqref{Eq: decomposition in eta-free elements - recursive definition}
		\\
		&=\eta_{\Lambda\setminus X}(A_\Lambda)
		-\sum_{Y\subset X}
		\sum_{Z\subseteq Y}(-1)^{|Y\setminus Z|}\eta_{\Lambda\setminus Z}(A_\Lambda)
		&|Y|<N
		\\
		&=\eta_{\Lambda\setminus X}(A_\Lambda)
		-\sum_{Z\subset X}
		\Big(\sum_{Z\subseteq Y\subset X}(-1)^{|Y|}\Big)(-1)^{|Z|}\eta_{\Lambda\setminus Z}(A_\Lambda)\,.
	\end{align*}
    Using the standard combinatorial identity
    \begin{align*}
        \sum_{S \subseteq X} (-1)^{|S|} = 0
        \qquad
        \textrm{for }
        X\neq\varnothing\,,
    \end{align*}
    we find
	\begin{align*}
		\sum_{Z\subseteq Y\subset X}(-1)^{|Y|}
		&=\sum_{S\subset X\setminus Z}(-1)^{|S|+|Z|}
		&Y=S\cup Z\,,\,S\subset X\setminus Z
		\\
		&=(-1)^{|Z|}\Big(\cancel{\sum_{S\subseteq X\setminus Z}(-1)^{|S|}}
		-(-1)^{|X\setminus Z|}\Big)
		&X\setminus Z\neq\varnothing
		\\
		&=(-1)^{|X|+1}\,.
	\end{align*}
	Overall we have
	\begin{align*}
		\tilde{A}_X
		=\eta_{\Lambda\setminus X}(A_\Lambda)
		+\sum_{Z\subset X}
		(-1)^{|X\setminus Z|}\eta_{\Lambda\setminus Z}(A_\Lambda)
		=\sum_{Z\subseteq X}
		(-1)^{|X\setminus Z|}\eta_{\Lambda\setminus Z}(A_\Lambda)\,.
	\end{align*}
	By induction, the proof is complete.
\end{proof}

\begin{remark}\label{Rmk: decomposition in eta-free elements - improved}
	\noindent
	\begin{enumerate}[(i)]
		\item
		Proposition \ref{Prop: decomposition in eta-free elements} provides a useful estimate on the norm of $\tilde{A}_X$ in terms of $A_\Lambda$.
		In particular we have
		\begin{align}\label{Eq: traceless elements - norm bound - classical case}
			\|\tilde{A}_X\|
			\leq 2^{|X|}\|A_\Lambda\|\,.
		\end{align}
		
		\item
		For later convenience we comment on a refinement of Proposition \ref{Prop: decomposition in eta-free elements}.
		To this avail, let $n\in\mathbb{N}$ and $X_1,\ldots,X_n,\Lambda\Subset\Gamma$, $Z\subset\Lambda$.
		Let $A_1,\ldots,A_n,\tilde{A}_\Lambda\in \mathcal{A}_\Gamma$ be such that $A_\ell\in \mathcal{A}_{X_\ell}$ for all $\ell\in\{1,\ldots,n\}$ while $\tilde{A}_\Lambda\in\tilde{\mathcal{A}}_\Lambda$.
		We would like to consider the decomposition \eqref{Eq: decomposition in eta-free elements} for the element
		\begin{align*}
			A_{S_n^\Lambda}
			:=A_1\cdots A_n\tilde{A}_\Lambda
			\in \mathcal{A}_{S_n^\Lambda}\,,
			\qquad
			S_n^\Lambda:=X_1\cup\ldots X_n\cup\Lambda\,.
		\end{align*}
		Let $S_n:=X_1\cup\ldots\cup X_n$ and set $\Lambda_n:=\Lambda\cap S_n^c$.
		(It may happen that $\Lambda_n=\varnothing$: This does not change the conclusions of the forthcoming discussion.)
		Then $\eta_x(A_{S_n^\Lambda})=0$ for $x\in\Lambda_n$.
		In other words, we already know that $A_{S_n^\Lambda}$ is $\eta$-free at points of $\Lambda$.
		This leads to a refined version of \eqref{Eq: decomposition in eta-free elements}, in particular
		\begin{align}\label{Eq: decomposition in eta-free elements - improved}
			A_{S_n^\Lambda}
			=\sum_{X\subseteq S_n}\tilde{A}_{X\cup\Lambda_n}\,,
		\end{align}
		where the coefficients are recursively defined by
		\begin{multline}\label{Eq: decomposition in eta-free elements - recursive definition - improved}
			\tilde{A}_{\Lambda_n}
			:=\eta_{S_n^\Lambda\setminus\Lambda_n}(A_{S_n^\Lambda})
			\\
			\tilde{A}_{X\cup\Lambda_n}
			:=\eta_{S_n^\Lambda\setminus(X\cup\Lambda_n)}(A_{S_n^\Lambda})
			-\sum_{Y\subset X}\tilde{A}_{Y\cup\Lambda_n}
			\qquad
			\forall X\subseteq S_n\,.
		\end{multline}
		Notice that the above sum over $Y$ is on subsets of $X\subseteq S_n$.
		Proceeding as in the proof of Proposition \ref{Prop: decomposition in eta-free elements} we have
		\begin{align}\label{Eq: decomposition in eta-free elements - explicit form - improved}
			\tilde{A}_{X\cup\Lambda_n}
			=\sum_{Y\subseteq X}(-1)^{|X\setminus Y|}\eta_{S_n^\Lambda\setminus(Y\cup\Lambda_n)}(A_{S_n^\Lambda})\,.
		\end{align}
		The main point is that equations \eqref{Eq: decomposition in eta-free elements - improved}-\eqref{Eq: decomposition in eta-free elements - explicit form - improved} contain sums over subsets $X$ of $S_n$ rather than of $S_n^\Lambda$ ---crucially, the size of the former set does not depend on $\Lambda$.
		This leads to a finer estimate on the norm of $\tilde{A}_{X\cup\Lambda_n}$:
		\begin{align}\label{Eq: eta-free elements - norm bound - improved}
			\|\tilde{A}_{X\cup\Lambda_n}\|
			\leq 2^{|X|}\|A_{S_n^\Lambda}\|\,.
		\end{align}
	\end{enumerate}
\end{remark}

\subsection{Proof of Theorem \ref{Thm: uniqueness for KMS quantum states}}
\label{Subsec: proof of uniqueness for quantum KMS states}

This section is devoted to the proof of Theorem \ref{Thm: uniqueness for KMS quantum states}.
This also covers the proof of Theorem \ref{Thm: uniqueness for KMS quantum states - with commutation} which is technically easier on account of assumption \eqref{Eq: assumption on potential - commutation}, \textit{cf.} Remark \ref{Rmk: proof of uniqueness for KMS quantum states - with commutation}.
The strategy of the proof is conceptually similar to \cite{Bratteli_Robinson_97,Drago_Pettinari_Van_de_Ven_2025,Frohlich_Ueltschi_2015} and requires a few preparations.

\bigskip

As anticipated, we will use the results described in Proposition \ref{Prop: decomposition in eta-free elements} for the specific choice of $\eta_x$ given by the unique $(\beta,\tau^\Psi)$-KMS state on $\mathcal{A}_x$:
\begin{align*}
	\eta_x(A_x):=\frac{\operatorname{tr}(e^{-\beta\Psi_x }A_x)}{\operatorname{tr}(e^{-\beta\Psi_x}A_x)}\,.
\end{align*}
In the forthcoming proof it will be convenient to recall that the normalized trace on $\mathcal{A}_x$ can be expressed as an average over the compact group $\mathcal{U}_x$ of unitary elements of $\mathcal{A}_x$:
\begin{align}\label{Eq: trace as average}
	\operatorname{tr}(A_x)
	=\int_{\mathcal{U}_x}U_xA_xU_x^*\mathrm{d}U_x\,,
\end{align}
where $\mathrm{d}U_x$ denotes the corresponding normalized Haar measure.
In particular this implies that
\begin{align}\label{Eq: Gibbs state as average}
	\eta_x(A_x)
	=\fint_{\mathcal{U}_x}U_xe^{-\beta\Psi_x}A_xU_x^*\mathrm{d}U_x\,,
\end{align}
where $\fint$ indicates that the integral is normalized so that $\eta_x(1_x)=1$.

\bigskip

For later convenience we introduce a suitable Banach space built out of the spaces of $\eta$-free elements.
In particular we consider
\begin{align}\label{Eq: Banach space}
	\underline{\mathsf{X}}
	:=\{\underline{f}=(f_\Lambda)_{\Lambda\Subset\Gamma}\mid
	f_\Lambda\colon\tilde{\mathcal{A}}_\Lambda\to\mathbb{C}\,,\,
	f_\Lambda\textrm{ linear, }
	\|\underline{f}\|_{\underline{\mathsf{X}}}
	:=\sup_{\Lambda\Subset\Gamma}
	\sup_{\tilde{A}_\Lambda\in\tilde{\mathcal{A}}_\Lambda}
	\frac{|f_\Lambda(\tilde{A}_\Lambda)|}{\|\tilde{A}_\Lambda\|}
	<+\infty
	\}\,.
\end{align}
Notice that any $\omega_\Gamma\in S(\mathcal{A}_\Gamma)$ defines an element $\underline{\omega}_\Gamma\in\underline{\mathsf{X}}$ by $ \underline{\omega}_\Gamma(\tilde{A}_\Lambda):=\omega_\Gamma(\tilde{A}_\Lambda)$.
In particular $\|\underline{\omega}_\Gamma\|_{\underline{\mathsf{X}}}=1$.
Moreover, the map
$
S(\mathcal{A}_\Gamma)
\ni\omega_\Gamma
\mapsto\underline{\omega}_\Gamma
\in\underline{\mathsf{X}}
$
is injective because any state $\omega\in S(\mathcal{A}_\Gamma)$ is uniquely determined by its value on the subspaces $\tilde{\mathcal{A}}_\Lambda$, $\Lambda\Subset\Gamma$.
Indeed, Proposition \ref{Prop: decomposition in eta-free elements} ensures that, by taking linear combinations of elements from the latter spaces, one can generate the dense $*$-subalgebra $\dot{\mathcal{A}}_\Gamma$.

\bigskip

We now begin the proof of Theorem \ref{Thm: uniqueness for KMS quantum states}.
Following \cite{Bratteli_Robinson_97,Drago_Pettinari_Van_de_Ven_2025,Frohlich_Ueltschi_2015}, we will show that, given a $(\beta,\tau^\Gamma)$-KMS quantum state $\omega_\Gamma^\beta\in S(\mathcal{A}_\Gamma)$, the corresponding element $\underline{\omega}_\Gamma^\beta\in\underline{\mathsf{X}}$ satisfies a linear equation whose solution is unique for sufficiently small $\beta$ --- 
more precisely, whenever condition \eqref{Eq: uniqueness for KMS quantum states - condition for uniqueness} is fulfilled.

To achieve this, we consider an ordering on $\Gamma$, which is only needed to select a preferred point for each $\Lambda \Subset \Gamma$ by setting $x:=\min_{y\in\Lambda}y$.
Given a $(\beta,\tau^\Gamma)$-KMS state $\omega_\Gamma^\beta\in S(\mathcal{A}_\Gamma)$ we consider $\tilde{A}_\Lambda\in\tilde{\mathcal{A}}_\Lambda$, for a non-empty $\Lambda\Subset\Gamma$.
Then Equation \eqref{Eq: Gibbs state as average} implies
\begin{align*}
	0=\omega_\Gamma^\beta(\eta_x(\tilde{A}_\Lambda))
	=\fint_{\mathcal{U}_x}\omega_\Gamma^\beta(U_xe^{-\beta\Psi_x}\tilde{A}_\Lambda U_x^*)\mathrm{d}U_x\,.
\end{align*}
By Proposition \ref{Prop: dynamics on quasi-local algebra - existence with strengthened assumptions}, $U_x e^{-\beta\Psi_x}$ is analytic for $\tau^\Gamma$. 
Moreover, condition \eqref{Eq: uniqueness for KMS quantum states - condition for uniqueness} implies $2\beta\|\bar{\Phi}\|_{\varepsilon,2\beta}<\varepsilon$, therefore, Remark \ref{Rmk: Dyson series - analytic continuation} applies.
We obtain
\begin{align*}
	\tau^\Gamma_{i\beta}(U_xe^{-\beta\Psi_x})
	=\sum_{n\geq 0}(-1)^n
	\int_0^\beta\mathrm{d}s_1
	\cdots
	\int_0^{s_{n-1}}\mathrm{d}s_n
	\sum_{X_1,\ldots,X_n:x}
	\operatorname{ad}_{\tau^\Psi_{is_1}\bar{\Phi}_{X_1}}
	\cdots
	\operatorname{ad}_{\tau^\Psi_{is_n}\bar{\Phi}_{X_n}}
	(e^{-\beta\Psi_x}U_x)\,,
\end{align*}
where we also used $\tau^\Psi_{i\beta}(U_xe^{-\beta\Psi_x})=e^{-\beta\Psi_x}U_x$.
Thus, the $(\beta,\tau^\Gamma)$-KMS condition implies
\begin{align*}
	0&=\fint_{\mathcal{U}_x}\omega_\Gamma^\beta\Big(\tilde{A}_\Lambda U_x^*\tau^\Gamma_{i\beta}(U_xe^{-\beta\Psi_x})\Big)\mathrm{d}U_x
	\\
	&=\sum_{n\geq 0}(-1)^n
	\int_0^\beta\mathrm{d}s_1
	\cdots
	\int_0^{s_{n-1}}\mathrm{d}s_n
	\sum_{X_1,\ldots,X_n:x}
	\omega_\Gamma^\beta\Big(
	\tilde{A}_\Lambda\fint_{\mathcal{U}_x} U_x^*\operatorname{ad}_{\tau^\Psi_{is_1}\bar{\Phi}_{X_1}}
	\cdots
	\operatorname{ad}_{\tau^\Psi_{is_n}\bar{\Phi}_{X_n}}
	(e^{-\beta\Psi_x}U_x)\mathrm{d}U_x\Big)\,.
\end{align*}
The term for $n=0$ in the above series leads to a contribution
\[
	\omega_\Gamma^\beta\Big(
	\tilde{A}_\Lambda\fint_{\mathcal{U}_x} U_x^* e^{-\beta\Psi_x}U_x\mathrm{d}U_x\Big)
	=\omega_\Gamma^\beta(\tilde{A}_\Lambda)\,,
\]
where we recall that the integral over $\mathcal{U}_x$ is normalized according to equation \eqref{Eq: Gibbs state as average}.
Rearranging the above equality leads to
\begin{multline}\label{Eq: uniqueness for KMS quantum states - linear equation - not explicit}
	\omega_\Gamma^\beta(\tilde{A}_\Lambda)
	=\sum_{n\geq 1}(-1)^{n+1}
	\int_0^\beta\mathrm{d}s_1
	\cdots
	\int_0^{s_{n-1}}\mathrm{d}s_n
	\\
	\sum_{X_1,\ldots,X_n:x}
	\omega_\Gamma^\beta\Big(
	\tilde{A}_\Lambda\fint_{\mathcal{U}_x} U_x^*\operatorname{ad}_{\tau^\Psi_{is_1}\bar{\Phi}_{X_1}}
	\cdots
	\operatorname{ad}_{\tau^\Psi_{is_n}\bar{\Phi}_{X_n}}
	(e^{-\beta\Psi_x}U_x)\mathrm{d}U_x\Big)\,.
\end{multline}
Equation \eqref{Eq: uniqueness for KMS quantum states - linear equation - not explicit} can be interpreted as a linear equation within $\underline{\mathsf{X}}$.
Indeed, the left-hand side coincide with $\underline{\omega}_\Gamma^\beta(\tilde{A}_\Lambda)$.
For the right-hand side we recall Proposition \ref{Prop: decomposition in eta-free elements} and decompose
\begin{align*}
    \tilde{A}_\Lambda\fint_{\mathcal{U}_x} U_x^*\operatorname{ad}_{\tau^\Psi_{is_1}\bar{\Phi}_{X_1}}
	\cdots
	\operatorname{ad}_{\tau^\Psi_{is_n}\bar{\Phi}_{X_n}}
	(e^{-\beta\Psi_x}U_x)\mathrm{d}U_x
	=\sum_{X\subseteq S_n}
	\tilde{B}_{X\cup\Lambda_n}\,,
\end{align*}
where we made use of Equation \eqref{Eq: decomposition in eta-free elements - improved} where $S_n:=X_1\cup\ldots\cup X_n$.
Let now $L_\beta\colon\underline{\mathsf{X}}\to\underline{\mathsf{X}}$ be defined by $(L_\beta\underline{f})_\varnothing:=0$ while, for non-empty $\Lambda\Subset\Gamma$,
\begin{align*}
	(L_\beta\underline{f})_\Lambda(\tilde{A}_\Lambda)
	:=\sum_{n\geq 1}(-1)^{n+1}
	\int_0^\beta\mathrm{d}s_1
	\cdots
	\int_0^{s_{n-1}}\mathrm{d}s_n
	\sum_{X_1,\ldots,X_n:x}
	\sum_{X\subseteq S_n}
	f_{X\cup\Lambda_n}(
	\tilde{B}_{X\cup\Lambda_n})\,.
\end{align*}
With this choice we find that, for any $(\beta,\tau^\Gamma)$-KMS state $\omega_\Gamma^\beta$ on $\mathcal{A}_\Gamma$, the corresponding element $\underline{\omega}_\Gamma^\beta\in\underline{\mathsf{X}}$ solves
\begin{align}\label{Eq: uniqueness for KMS quantum states - linear equation}
	(I-L_\beta)\underline{\omega}_\Gamma^\beta=\underline{\delta}\,,
\end{align}
where $\underline{\delta}\in\underline{\mathsf{X}}$ is defined by
\begin{align}\label{Eq: uniqueness for KMS quantum states - linear equation - source term}
	\underline{\delta}_X(\tilde{A}_X)
	=\begin{dcases}
		\tilde{A}_\varnothing&X=\varnothing
		\\
		0&\textrm{otherwise}
	\end{dcases}\,.
\end{align}
We now prove that $L_\beta \colon \underline{\mathsf{X}} \to \underline{\mathsf{X}}$ is a well-defined bounded linear operator, and moreover that $\|L_\beta\| < 1$ provided condition \eqref{Eq: uniqueness for KMS quantum states - condition for uniqueness} is met.  
This guarantees uniqueness of the solution to Equation \eqref{Eq: uniqueness for KMS quantum states - linear equation}.  
Since every $(\beta,\tau^\Gamma)$-KMS state on $\mathcal{A}_\Gamma$ gives rise to such a solution, and since the map
\begin{align*}
	S(\mathcal{A}_\Gamma)\ni\omega_\Gamma
	\mapsto\underline{\omega}_\Gamma \in \underline{\mathsf{X}}\,,
\end{align*}
is injective, it follows that $\omega_\Gamma^\beta$ is unique.

To prove the well-definiteness of $L_\beta$ we first recall that, by Remark \ref{Rmk: decomposition in eta-free elements - improved}, we have
\begin{align*}
	\|\tilde{B}_{X\cup\Lambda_n}\|
	\leq 2^{|X|}\|\tilde{A}_\Lambda\|\Big\|
	\fint_{\mathcal{U}_x} U_x^*\operatorname{ad}_{\tau^\Psi_{is_1}\bar{\Phi}_{X_1}}
	\cdots
	\operatorname{ad}_{\tau^\Psi_{is_n}\bar{\Phi}_{X_n}}
	(e^{-\beta\Psi_x}U_x)\mathrm{d}U_x
	\Big\|\,.
\end{align*}
Moreover, by expanding the nested commutators and invoking once again Equation \eqref{Eq: Gibbs state as average}, we obtain
\begin{multline*}
	\fint_{\mathcal{U}_x}
	U_x^*
	(\operatorname{ad}_{\tau^\Psi_{is_1}\bar{\Phi}_{X_1}}
	\cdots
	\operatorname{ad}_{\tau_{is_n}\bar{\Phi}_{X_n}})(e^{-\beta\Psi_x}U_x)\mathrm{d}U_x
	\\
	=\sum_{\wp\in\{0,1\}^n}
	(-1)^{n-|\wp|}
	\eta_x\Big(\tau^\Psi_{is_1}(\bar{\Phi}_{X_1})^{(\wp_1)}
	\cdots
	\tau^\Psi_{i s_n}(\bar{\Phi}_{X_n})^{(\wp_n)}\Big)
	\tau^\Psi_{is_n}(\bar{\Phi}_{X_n})^{(1-\wp_n)}
	\dots
	\tau^\Psi_{i s_1}\bar{\Phi}_{X_1}^{(1-\wp_1)}\,,
\end{multline*}
where $\wp=(\wp_1,\ldots,\wp_n)$, $|\wp|:=\wp_1+\ldots+\wp_n$ while
\begin{align*}
	\tau^\Psi_{i s_\ell}(\bar{\Phi}_{X_\ell})^{(\wp_\ell)}
	:=\begin{dcases}
	\tau^\Psi_{i s_\ell}(\bar{\Phi}_{X_\ell})
	&\wp_\ell=1
	\\
	1
	&\wp_\ell=0
	\end{dcases}\,.
\end{align*}
This leads to
\begin{align*}
	\Big\|
	\fint_{\mathcal{U}_x} U_x^*\operatorname{ad}_{\tau^\Psi_{is_1}\bar{\Phi}_{X_1}}
	\cdots
	\operatorname{ad}_{\tau^\Psi_{is_n}\bar{\Phi}_{X_n}}
	(e^{-\beta\Psi_x}U_x)\mathrm{d}U_x
	\Big\|
	\leq 2^n\prod_{\ell=1}^ne^{2\beta\|\Psi\|_{X_\ell}}\|\bar{\Phi}_{X_\ell}\|\,,
\end{align*}
so that overall
\[
	\|\tilde{B}_{X\cup\Lambda_n}\|
	\leq 2^{|X|}\|\tilde{A}_\Lambda\| 2^n\prod_{\ell=1}^n e^{2\beta\|\Psi\|_{X_\ell}}\|\bar{\Phi}_{X_\ell}\|\,.
\]
The above estimate and Lemma \ref{Lem: sum on subsets estimate} implies that
\begin{align*}
	|(L_\beta\underline{f})_\Lambda(\tilde{A}_\Lambda)|
	&\leq\|\underline{f}\|\|\tilde{A}_\Lambda\|
	\sum_{n\geq 1}\frac{(2\beta)^n}{n!}
	\sum_{X_1,\ldots,X_n:x}
	\prod_{\ell=1}^n 3^{|X_\ell|}e^{2\beta\|\Psi\|_{X_\ell}}\|\bar{\Phi}_{X_\ell}\|
	\\
	&\leq\|\underline{f}\|\|\tilde{A}_\Lambda\|e^{\varepsilon}
	\sum_{n\geq 1}\Big(\frac{6\beta}{\varepsilon}
	\|\bar{\Phi}\|_{\varepsilon+\log 3,2\beta}\Big)^n
	\\
	&\leq\|\underline{f}\|\|\tilde{A}_\Lambda\|e^{\varepsilon}
	\frac{\frac{6\beta}{\varepsilon}\|\bar{\Phi}\|_{\varepsilon+\log 3,2\beta}}{1-\frac{6\beta}{\varepsilon}\|\bar{\Phi}\|_{\varepsilon+\log 3,2\beta}}\,,
\end{align*}
where the convergence of the geometric series is granted by condition \eqref{Eq: uniqueness for KMS quantum states - condition for uniqueness}.
It follows that $L_\beta\in\mathcal{B}(\underline{\mathsf{X}})$ with
\begin{align*}
	\|L_\beta\|
	\leq
	e^{\varepsilon}
	\frac{\frac{6\beta}{\varepsilon}\|\bar{\Phi}\|_{\varepsilon+\log 3,2\beta}}{1-\frac{6\beta}{\varepsilon}\|\bar{\Phi}\|_{\varepsilon+\log 3,2\beta}}
	<1\,,
\end{align*}
where the last equality is equivalent to condition \eqref{Eq: uniqueness for KMS quantum states - condition for uniqueness}.
This completes the proof.

\begin{remark}\label{Rmk: proof of uniqueness for KMS quantum states - with commutation}
	The proof of Theorem \ref{Thm: uniqueness for KMS quantum states - with commutation} follows the same steps as that of Theorem \ref{Thm: uniqueness for KMS quantum states}.
	However, assumption \eqref{Eq: assumption on potential - commutation} yields a technical simplification, since
	\[
	\tau^\Gamma_{i\beta}(U_x e^{-\beta\Psi_x})
	=
	e^{-\beta\Psi_x}\,\bar{\tau}^\Gamma_{i\beta}(U_x)\,.
	\]
	This identity leads to a simplified version of equation \eqref{Eq: uniqueness for KMS quantum states - linear equation - not explicit}, namely
	\begin{multline*}
	\omega_\Gamma^\beta(\tilde{A}_\Lambda)
	=
	\sum_{n\geq 1} (-1)^{n+1}
	\int_0^\beta \mathrm{d}s_1
	\cdots
	\int_0^{s_{n-1}} \mathrm{d}s_n
	\\
	\sum_{X_1,\ldots,X_n:x}
	\omega_\Gamma^\beta\Big(
	\tilde{A}_\Lambda
	\fint_{\mathcal{U}_x}
	U_x^*\operatorname{ad}_{\bar{\Phi}_{X_1}}
	\cdots
	\operatorname{ad}_{\bar{\Phi}_{X_n}}
	\big(e^{-\beta\Psi_x}U_x\big)
	\,\mathrm{d}U_x
	\Big)\,.
	\end{multline*}
	This expression yields a linear equation in $\underline{\mathsf{X}}$ of the form \eqref{Eq: uniqueness for KMS quantum states - linear equation}, where the operator norm of $L_\beta$ can be estimated as
	\[
	\|L_\beta\|
	\leq
	e^\varepsilon
	\frac{\frac{6\beta}{\varepsilon}\|\bar{\Phi}\|_{\varepsilon+\log 3}}
	{1 - \frac{6\beta}{\varepsilon}\|\bar{\Phi}\|_{\varepsilon+\log 3}}\,.
	\]
	In particular, this estimate avoids the need to invoke the strengthened norm \eqref{Eq: assumption on potential - norm on potentials strengthened}.
	The remainder of the proof then proceeds unchanged.
\end{remark}


\section{Consequences of Theorem \ref{Thm: uniqueness for KMS quantum states}}
\label{Sec: consequences of (main) Theorem}

In this section we present several direct consequences of Theorem \ref{Thm: uniqueness for KMS quantum states}.
Specifically, in Section \ref{Subsec: Comparison of the subcritical inverse temperatures} we compare our subcritical inverse temperature $\beta_{\textsc{u}}$ with those already available in the literature.
It is worth emphasizing that our result does not necessarily provide a sharper value, but it does yield a value that is uniform with respect to the dimension of the single-site Hilbert space $\mathcal{H}_x$.
Section \ref{Subsec: Uniqueness result for classical KMS states} is devoted to the proof of Theorem \ref{Thm: uniqueness for KMS classical states}, which serves as the classical analogue of Theorem \ref{Thm: uniqueness for KMS quantum states}.
This demonstrates that our proof is sufficiently robust to carry over to the classical setting.
Finally, in Section \ref{Subsec: common subcritical region for quantum and classical systems}, following \cite{Drago_Pettinari_Van_de_Ven_2025}, we briefly discuss the implications of this result.
In particular, the combination of Theorems \ref{Thm: uniqueness for KMS quantum states}-\ref{Thm: uniqueness for KMS classical states} ensures the existence of a common high-temperature regime in which both classical and quantum phase transitions are absent.

\subsection{Comparison of the subcritical inverse temperatures}
\label{Subsec: Comparison of the subcritical inverse temperatures}

The novelty of Theorems \ref{Thm: uniqueness for KMS quantum states - with commutation}-\ref{Thm: uniqueness for KMS quantum states} lies in the fact that the value of the subcritical inverse temperature $\beta_{\textsc{u}}$ is independent of the dimension of the Hilbert spaces $\mathcal{H}_x$.
Considering in particular Theorem \ref{Thm: uniqueness for KMS quantum states - with commutation}, it is possible to compare our value for $\beta_{\textsc{u}}$ with those already present in the literature.
We will perform this comparison with two precise models, namely the anisotropic Heisenberg model, and the Ising model with a staggered external magnetic field.

\bigskip

From Theorem \ref{Thm: uniqueness for KMS quantum states - with commutation} we obtain the following "optimal" inverse subcritical temperature $\beta_{\textsc{u}}$:
\begin{align}\label{Eq: estimate of subcritical inverse temperature - commutation}
	\beta_{\textsc{u}}\|\bar{\Phi}\|_{\varepsilon+\log 3}
	:=\frac{1}{6}\frac{\varepsilon}{1+e^\varepsilon}\,.
\end{align}
We now compare the latter with the estimates provided by \cite[Prop. 6.2.45, Thm. 6.2.46]{Bratteli_Robinson_97}.
As we will see, the subcritical inverse temperature provided by Theorem \ref{Thm: uniqueness for KMS quantum states} is larger than the one obtained by applying the results of \cite{Bratteli_Robinson_97}.
In particular, our results provide a larger regime of uniqueness for KMS quantum states.

For simplicity, we specialize to the case $\Gamma = \mathbb{Z}^{\nu}$ and $\mathcal{H}_x=\mathbb{C}^{2j+1}$, $j\in\mathbb{Z}_+/2$.
Additionally, we will denote by $\Gamma_b$ the set of bonds, i.e. pairs of points $\{x,y\}$ in $\Gamma$ with graph distance $\text{d}(x,y)$ equal to 1.
The \textbf{quantum anisotropic Heisenberg model} is described by the interaction potential \begin{equation}
    \Phi_\Lambda:=\begin{dcases}
    	-J(x,y)(\delta (S^1_x S^1_y + S^2_x S^2_y) + S^3_x S^3_y) &\Lambda= \{x,y\} \in \Gamma_b \\
    	0 & \text{otherwise}
    \end{dcases},
\end{equation} with $$J\colon \Gamma_b \to \mathbb{R},\qquad\sup_{\{x,y\}\in \Gamma_b} |J(x,y)| <+\infty,$$
and $(S^1_x,S^2_x,S^3_x)$ the spin operators at site $x\in\Gamma$, with spin $j$.
This potential does not have a single-site component, so that we have $\Phi = \overline{\Phi}$ and we can evaluate the norm in Theorem \ref{Thm: uniqueness for KMS quantum states - with commutation} as \begin{equation}
	\| \Phi\|_{\varepsilon +\log3}= 3 e^{\varepsilon} \sup_{x\in \Gamma}\sum_{y:\{y,x\}\in\Gamma_b}|J(x,y)| \| \delta(S_x^1 S_y^1 + S_x^2 S_y^2) + S^3_x S^3_y\|.
\end{equation}
Hence, the corresponding inverse critical temperature reads
\begin{equation}
\beta_{\textsc{u}} = \frac{1}{18 }\frac{1}{\sup_{x\in \Gamma}\sum_{y:\{y,x\}\in\Gamma_b}|J(x,y)| \| \delta(S_x^1 S_y^1 + S_x^2 S_y^2) + S^3_x S^3_y\|} \frac{\varepsilon e^{-\varepsilon}}{1+e^\varepsilon}\,.
\end{equation}
The optimal value $\bar{\varepsilon}$ for $\varepsilon$ can be determined numerically as $\bar{\varepsilon} \approx 0.607$, maximizing \begin{equation*}
\frac{\bar{\varepsilon}e^{-\bar{\varepsilon}}}{1+ e^{\bar{\varepsilon}}} \approx 0.117\,.
\end{equation*}

In comparison, the "optimal" inverse temperature $\beta_{\textsc{u}}^{\textsc{br}}$ obtained by using \cite[Prop. 6.2.45]{Bratteli_Robinson_97} is given by
\begin{equation}
    \beta_{\textsc{u}}^{\textsc{br}}
    =\frac{1}{2(2j+1)^2}\frac{1}{\sup_{x\in \Gamma}\sum_{y:\{y,x\}\in\Gamma_b}|J(x,y)| \| \delta(S_x^1 S_y^1 + S_x^2 S_y^2) + S^3_x S^3_y\|}
    \frac{\varepsilon e^{-\varepsilon}}{1+ e^{\varepsilon} \frac{(2j+1)^3}{2j}}\,.
\end{equation}
In this case the optimal value $\bar{\varepsilon}_j$ for $\varepsilon$ depends on $j$: for $j = 1/2$ we have $\bar{\varepsilon}_{1/2}\approx 0.518 $ and for large values of the spin $\bar{\varepsilon}_j$ approach $1/2$ from above.
The ratio between the two subcritical inverse temperatures is then 
\begin{equation}
\frac{\beta_{\textsc{u}}^{\textsc{br}}}{\beta_{\textsc{u}}} = \frac{9}{(2j+1)^2}\frac{\bar{\varepsilon}_j e^{-\bar{\varepsilon}_j }  }{1+ e^{\bar{\varepsilon}_j}\frac{(2j+1)^3}{2j}}
\frac{1+e^{\bar{\varepsilon}}}{\bar{\varepsilon}e^{-\bar{\varepsilon}}}\,.
\end{equation}
Fixing $j = 1/2$, this yields \begin{equation}
\frac{\beta_{\textsc{u}}^{\textsc{br}}}{\beta_{\textsc{u}}} \approx 0.412,
\end{equation} which shows how Theorem \ref{Thm: uniqueness for KMS quantum states - with commutation} increases the inverse critical temperature predicted by \cite[Prop. 6.2.45]{Bratteli_Robinson_97} by more than a factor of two. Furthermore, since the estimate provided by our theorem is independent of the single-site dimension $2j+1$, its accuracy remains unaltered as $j$ is increased, whereas the estimate of \cite[Prop. 6.2.45]{Bratteli_Robinson_97} gets worse.

\bigskip

We now turn to the \textbf{quantum Ising model with a staggered magnetic field}, defined by \begin{equation}
	\Phi_\Lambda := \begin{dcases}
		J S^3_x S^3_y & \Lambda = \{x,y\}\in \Gamma_b\\
		(-1)^{\text{d}(x,0)} B S^3_x & \Lambda = \{x\} \\
		0 & \text{otherwise}
	\end{dcases},
\end{equation} where $B>0$ represents an external magnetic field applied to the systems. Applying directly Theorem \ref{Thm: uniqueness for KMS quantum states} or the result in \cite[Prop. 6.2.45]{Bratteli_Robinson_97} would lead to a dependence on the external magnetic field $B$ for the subcritical inverse temperature. However, since the single-site potential commutes with the interaction, we can employ Theorem \ref{Thm: uniqueness for KMS quantum states - with commutation} and Proposition \cite[Prop. 6.2.46]{Bratteli_Robinson_97}  to subtract the commuting potential from $\Phi$.
Thus,  Theorem \ref{Thm: uniqueness for KMS quantum states - with commutation} applies and we have the following subcritical inverse temperature
\begin{equation}
    \beta_{\textsc{u}} = \frac{1}{36 \nu |J|}\frac{\varepsilon e^{-\varepsilon} }{1+ e^{\varepsilon}}\,,
\end{equation}
which is maximized for $\bar{\varepsilon}\approx 0.607$.
The "optimal" subcritical inverse temperature provided by \cite[Prop. 6.2.46]{Bratteli_Robinson_97} reads
\begin{equation}
    \beta_{\textsc{u}}^{\textsc{br}}
    =\frac{1}{8 \nu |J| (2j+1)^3}\frac{\varepsilon e^{-\varepsilon}}{1 + 2(2j+1)^4 e^{\varepsilon}}\,.
\end{equation}
If we fix $j=1/2$, this value is maximized at $\bar{\varepsilon}_{1/2} \approx 0.505$. We can again compare the two inverse critical temperatures by evaluating their ratio \begin{equation}
    \frac{\beta^{\textsc{br}}_\textsc{u}}{\beta_{\textsc{u}}} = \frac{9}{2(2j+1)^3} \frac{\bar{\varepsilon}_j e^{-\bar{\varepsilon}_j }  }{1+ e^{\bar{\varepsilon}_j}2(2j+1)^4 }
    \frac{1+e^{\bar{\varepsilon}}}{\bar{\varepsilon}e^{-\bar{\varepsilon}}}
    \stackrel{j=\sfrac{1}{2}}{\approx}
    0.027\,,
\end{equation}
where the final estimate has been computed for $j=1/2$.
Once again, we find $\beta_{\textsc{u}}>\beta_{\textsc{u}}^{\textsc{br}}$.

\subsection{Uniqueness result for classical KMS states}
\label{Subsec: Uniqueness result for classical KMS states}

The goal of this section is to prove Theorem \ref{Thm: uniqueness for KMS classical states}, which is the classical counterpart of Theorem \ref{Thm: uniqueness for KMS quantum states}.
Barring minor differences related to the different classical setting, the proof is extremely similar to the quantum counterpart, the assumptions being slightly weaker.
We will also compare Theorems \ref{Thm: uniqueness for KMS quantum states}-\ref{Thm: uniqueness for KMS classical states} in relation with the strict deformation quantization discussed in \cite{Drago_Pettinari_Van_de_Ven_2025,Landsman_1998_I} and prove Corollary \ref{Cor: absence of classical and quantum phase transitions}.

\bigskip

We will now describe the main ingredients to state Theorem \ref{Thm: uniqueness for KMS classical states} precisely.
In what follows, we will be particularly interested in $\mathbb{S}^2$-valued classical spin system.
We will briefly summarize the main features, referring to \cite{Drago_Pettinari_Van_de_Ven_2025,Friedli_Velenik_2017} for a more detailed discussion.
We will consider the configuration space $\Omega$ which consists of functions $\Gamma\to\mathbb{S}^2$.
Thus, $\Omega=(\mathbb{S}^2)^\Gamma$, which is compact in the product topology.
The corresponding algebra of quasi-local observables is then defined by $\mathscr{A}_\Gamma:=C(\Omega)$.
The latter $C^*$-algebra can be seen as the direct limit of the net of $C^*$-algebras $\{\mathscr{A}_\Lambda\}_{\Lambda\Subset\Gamma}$, where $\mathscr{A}_\Lambda:=C(\Omega_\Lambda)$ with $\Omega_\Lambda=(\mathbb{S}^2)^\Lambda$ is the configuration space localized at $\Lambda$.
As in section \ref{Sec: Uniqueness result for quantum KMS states} we will not write explicitly the natural inclusions $\mathscr{A}_\Lambda\hookrightarrow\mathscr{A}_{\Lambda'}$, $\Lambda\subset\Lambda'$.

For each $\Lambda\Subset\Gamma$, $\Omega_\Lambda$ is a compact symplectic manifold once equipped with the tensor product of the symplectic form on $\mathbb{S}^2$.
Thus, we may consider the corresponding Poisson bracket $\{a_\Lambda,a'_\Lambda\}_\Lambda$ for any pairs in $a_\Lambda,a'_\Lambda\in\dot{\mathscr{A}}_\Lambda:=C^\infty(\Omega_\Lambda)$.
This implies that $\mathscr{A}_\Lambda$ is a $C^*$-Poisson algebra.
Similarly, $\mathscr{A}_\Gamma$ is a $C^*$-Poisson algebra once equipped with the Poisson bracket
\begin{align}
	\{\;,\;\}
	\colon\dot{\mathscr{A}}_\Gamma
	\times\dot{\mathscr{A}}_\Gamma
	\to\dot{\mathscr{A}}_\Gamma
	\qquad
	\dot{\mathscr{A}}_\Gamma
	:=\bigcup_{\Lambda\Subset\Gamma}\dot{\mathscr{A}}_\Lambda\,,
	\qquad
	\{a_\Lambda,a_{\Lambda'}\}
	:=\{a_\Lambda,a_{\Lambda'}\}_{\Lambda\cap\Lambda'}\,.
\end{align}

Next, we move to the discussion on the classical KMS condition.
The latter stands as the classical counterparts of the quantum KMS condition discussed in Section \ref{Subsec: quantum quasi-local algebra}.
Similarly to the quantum setting described in the previous sections, we will consider a family $\varphi=\{\varphi_\Lambda\}_{\Lambda\Subset\Gamma}$ of self-adjoint potentials.
In particular we assume that $\varphi_\varnothing=0$, $\varphi_\Lambda=\varphi_\Lambda^*\in\dot{\mathscr{A}}_\Lambda$ for all $\Lambda\Subset\Gamma$.
We will also consider a decomposition similar to \eqref{Eq: decomposition of quantum potential}, that is, we will set
\begin{align}\label{Eq: decomposition of classical potential}
	\bar{\varphi}_\Lambda
	:=\begin{dcases}
	\varphi_\Lambda
	&|\Lambda|\geq 2
	\\
	0
	&|\Lambda|=1
	\end{dcases}\,,
	\qquad
	\psi_x:=\varphi_{\{x\}}\,.
\end{align}
Contrary to the KMS quantum condition, which focuses on the dynamics $\tau^\Gamma$ induced by the family of quantum potentials $\Phi$, the KMS classical condition involves only a derivation $\delta_\Gamma\colon\dot{\mathscr{A}}_\Gamma\to\mathscr{A}_\Gamma$ on the algebra of quasi-local observables $\mathscr{A}_\Gamma$.
This leads to a slightly different condition on the family of classical potential $\varphi$ compared to condition \eqref{Eq: assumption on potential - norm on potentials}.
In particular, in what follows we will assume that
\begin{align}\label{Eq: assumption on potential - C1 norm on potentials}
	\sup_{x\in\Gamma}
	\sum_{\Lambda\ni x}
	\|\bar{\varphi}_\Lambda\|_{C^1(\Omega_\Lambda)}
	<+\infty\,.
\end{align}
Compared with \eqref{Eq: assumption on potential - norm on potentials}, condition \eqref{Eq: assumption on potential - C1 norm on potentials} is slightly more permissive for what concern the decay properties of the norm of the potentials ---Roughly speaking, $\varepsilon=0$.
However, the norm involved in condition \eqref{Eq: assumption on potential - C1 norm on potentials} is the $C^1$-norm on $\Omega_\Lambda$ rather than the natural $C^0$-norm.
This is motivated by the very definition of the derivation $\delta_\Gamma$ associated with $\varphi$, which reads
\begin{align}
	\delta_\Gamma(a_\Lambda)
	:=\sum_{\substack{X\Subset\Gamma\\X\cap\Lambda\neq\varnothing}}
	\{a_\Lambda,\varphi_X\}
	\qquad
	\forall a_\Lambda\in\dot{\mathscr{A}}_\Lambda
	\subset\dot{\mathscr{A}}_\Gamma\,.
\end{align}
Notice that no assumptions are required on the single-site potentials $\{\psi_x\}_{x\in\Gamma}$, since their contribution can always be incorporated into the action of a local Hamiltonian $h^\Psi_\Lambda$:
\begin{align*}
	\delta_\Gamma(a_\Lambda)
	=\sum_{\substack{X\Subset\Gamma\\X\cap\Lambda\neq\varnothing}}
	\{a_\Lambda,\bar{\varphi}_X\}
	+\{a_\Lambda,h^\psi_\Lambda\}\,,
	\qquad
	h^\psi_\Lambda
	:=\sum_{x\in\Lambda}\psi_x\,.
\end{align*}
This is in accordance with the quantum counterpart $\Psi$ of the single-site potential, whose associated dynamics $\tau^\Psi$ is readily well-defined on $\mathscr{A}_\Gamma$.

For the rest of this section we will focus on KMS classical states on $\mathscr{A}_\Gamma$ relative to $\delta_\Gamma$.
As mentioned in the introduction, the latter are states on $\mathscr{A}_\Gamma$ fulfilling condition \eqref{Eq: classical KMS condition}.
We will refer to them as $\delta_\Gamma$-KMS for short or, if a specific inverse temperature $\beta\geq 0$ has been fixed, we will speak of $(\beta,\delta_\Gamma)$-KMS classical states.

The existence of $\delta_\Gamma$-KMS classical states can be proved by a weak* compactness argument.
In particular, one first consider $\delta_\Lambda$-KMS states on $\mathscr{A}_\Lambda$, $\Lambda\Subset\Gamma$, where $\delta_\Lambda:=\{\;,h_\Lambda\}$ is the Hamiltonian vector field associated with $h_\Lambda=\sum_{X\subseteq\Lambda}\varphi_X$.
According to \cite{Bordermann_Hartmann_Romer_Waldmann_98,Drago_Waldmann_2024}, for any $\beta\geq 0$, there exists a unique $\delta_\Lambda$-KMS classical state $\omega_\Lambda^\beta\in S(\mathscr{A}_\Lambda)$, which is given by the classical Gibbs state
\begin{align*}
	\omega_\Lambda^\beta(a_\Lambda)
	:=\frac{\int_{\Omega_\Lambda}e^{-\beta h_\Lambda}a_\Lambda \mathrm{d}\mu_\Lambda}{\int_{\Omega_\Lambda}e^{-\beta h_\Lambda}\mathrm{d}\mu_\Lambda}\,,
\end{align*}
where $\mu_\Lambda$ is the standard Liouville measure on $\Omega_\Lambda$.
At this stage one observes that, on account of condition \eqref{Eq: assumption on potential - C1 norm on potentials}, $\delta_\Gamma=\textrm{s-}\lim_{\Lambda\uparrow\Gamma}\delta_\Lambda$ while any weak* limit point of the sequence $(\omega_\Lambda^\beta)_{\Lambda\Subset\Gamma}$ (which exists by Banach-Alaoglu theorem) defines a $(\beta,\delta_\Gamma)$-KMS state on $\mathscr{A}_\Gamma$.

We will now discuss the proof of Theorem \ref{Thm: uniqueness for KMS classical states} which concerns with the uniqueness of $(\beta,\delta_\Gamma)$-KMS classical states for sufficently low inverse temperature $\beta$.
Remarkably, the latter essentially follows the same steps of the proof of Theorem \ref{Thm: uniqueness for KMS quantum states}.
In particular, it is proven that any $(\beta,\delta_\Gamma)$-KMS classical state identifies a solution to a linear equation in a suitable Banach space.
For sufficiently low $\beta$ the latter equation has a unique solution, leading to the uniqueness of $(\beta,\delta_\Gamma)$-KMS classical states.
The difference between the proof of theorems \ref{Thm: uniqueness for KMS quantum states}-\ref{Thm: uniqueness for KMS classical states} lies in the derivation of the linear equation ---where the KMS condition is used--- and the proof that such equation has a unique solution.

To begin with, we recall an invariance property of $\delta_\Gamma$-KMS classical states, \textit{cf.} \cite[Prop. 5]{Drago_Waldmann_2024}, which provides an "integrated" version of the KMS classical condition which is easier to deal with.
To set the stage, let $x\in\Gamma$ and $\varrho_x\in SO(3)$.
We define $\varrho_x\colon\Omega_\Gamma\to\Omega_\Gamma$ by setting
\begin{align*}
	(\varrho_x\sigma)(x):=
	\begin{dcases}
	\sigma(y)
	&y\neq x
	\\
	\sigma(\varrho_x^{-1}x)
	&y=x
	\end{dcases}
	\,.
\end{align*}
The corresponding pull-back action on $\mathscr{A}_\Gamma$ will be denoted by $\hat{\varrho}_x$.
Then if $\omega_\Gamma^\beta\in S(\mathscr{A}_\Gamma)$ is a $(\beta,\delta_\Gamma)$-KMS classical states and $a_\Lambda\in\mathscr{A}_\Lambda$, $\Lambda\Subset\Gamma$, it holds
\begin{align}\label{Eq: invariance of KMS classical states under G action}
	\omega_\Gamma^\beta(a_\Lambda)
	=\omega_\Gamma^\beta\Big(
	e^{\beta\sum_{X\ni x}(I-\hat{\varrho}_x)\varphi_X}
	\hat{\varrho}_xa_\Lambda
	\Big)\,.
\end{align}	
Next, we observe that the results of Section \ref{Subsec: decomposition in eta-free elements} generalizes immediately to the case of the classical quasi-local algebra $\mathscr{A}_\Gamma$.
Indeed, the specific structure of the algebras $\mathcal{A}_x$ is not essential in the proof of Proposition \ref{Prop: decomposition in eta-free elements} and Remark \ref{Rmk: decomposition in eta-free elements - improved}.
In fact, Proposition \ref{Prop: decomposition in eta-free elements} and Remark \ref{Rmk: decomposition in eta-free elements - improved} naturally extend to any inductive system of $C^*$-algebras $\{\mathfrak{A}_\Lambda\}_{\Lambda \Subset \Gamma}$.

For the classical quasi-local algebra $\mathscr{A}_\Gamma$ we choose $\eta_x$ as follows:
\begin{align*}
	\eta_x(a_x) := \fint_{\Omega_x} e^{-\beta\psi_x}a_x \,\mathrm{d}\mu\,,
\end{align*}
where $\mu$ is the standard Liouville measure on $\mathbb{S}^2$ while $\fint$ indicates normalization.
Notice that $\eta_x$ is nothing but the unique $(\beta,\delta_\psi)$-KMS classical state on $\mathscr{A}_x$ where $\delta_\psi(a_x):=\{a_x,\psi_x\}$.
This choice is mirroring the one already done in the quantum setting, where $\eta_x$ was chosen to be the unique $(\beta,\tau^\Psi)$-KMS quantum state on $\mathcal{A}_x$.

We will denote by $\tilde{\mathscr{A}}_\Lambda$ the corresponding spaces of $\eta$-free elements, \textit{cf.} Definition \ref{Eq: eta-free elements - definition}.
By definition, $a_\Lambda\in\tilde{\mathscr{A}}_\Lambda$ if and only if $\eta_x(a_\Lambda)=0$ for all $x\in\Lambda$.
We will denote by $\underline{\mathsf{X}}$ the corresponding Banach space defined as in \eqref{Eq: Banach space}.
Furthermore, similarly to \eqref{Eq: Gibbs state as average}, the action of $\eta_x$ can be conveniently be written as
\begin{align}\label{Eq: Gibbs state as average - classical case}
	\eta_x(a_x)
	=\fint_{SO(3)}\hat{\varrho}_x[e^{-\beta\psi_x}a_x]\mathrm{d}\varrho_x\,,
\end{align}
where $\mathrm{d}\varrho_x$ denotes the Haar measure of $SO(3)$.

\begin{proof}[Proof of Theorem \ref{Thm: uniqueness for KMS classical states}]
	Let $\omega_\Gamma^\beta\in S(\mathscr{A}_\Gamma)$ be a $(\beta,\delta_\Gamma)$-KMS classical state.
	For any $\Lambda\Subset\Gamma$ let $x:=\min_{y\in\Lambda}y$ where we considered an arbitrary but fixed ordering on $\Gamma$.
	For $\tilde{a}_\Lambda\in\tilde{\mathscr{A}}_\Lambda$ we have
	\begin{align*}
		0&=\omega_\Gamma^\beta\Big(
		\fint_{\Omega_x}e^{-\beta\psi_x}\tilde{a}_\Lambda\mathrm{d}\mu_x
		\Big)
		&\tilde{a}_\Lambda\in\tilde{\mathscr{A}}_\Lambda
		\\
		&=\fint_{SO(3)}\omega_\Gamma^\beta(\hat{\varrho}_x[e^{-\beta\psi_x}\tilde{a}_\Lambda])\mathrm{d}\varrho_x
		&\textrm{Eq. }\eqref{Eq: Gibbs state as average - classical case}
		\\
		&=\fint_{SO(3)}\omega_\Gamma^\beta(e^{-\beta\psi_x}\tilde{a}_\Lambda
		e^{\beta\sum_{X\ni x}(I-\hat{\varrho}_x)\varphi_X})
		\mathrm{d}\varrho_x
		&\textrm{Eq. }\eqref{Eq: invariance of KMS classical states under G action}
		\\
		&=\omega_\Gamma^\beta\Big(\tilde{a}_\Lambda\fint_{SO(3)}
		e^{\beta\sum_{X\ni x}(I-\hat{\varrho}_x)\bar{\varphi}_X}e^{-\beta\hat{\varrho}_x\psi_x}
		\mathrm{d}\varrho_x\Big)
		\\
		&=\omega_\Gamma^\beta\Big(\tilde{a}_\Lambda\fint_{SO(3)} (e^{\beta\sum_{X\ni x}(I-\hat{\varrho}_x)\bar{\varphi}_X}-1)e^{-\beta\hat{\varrho}_x\psi_x}
		\mathrm{d}\varrho_x\Big)
		+\omega_\Gamma^\beta(\tilde{a}_\Lambda)\,,
	\end{align*}
	where in the last equality we used that the integral over $SO(3)$ is normalized according to equation \eqref{Eq: Gibbs state as average - classical case}.
	Overall we find
	\begin{align*}
		\omega_\Gamma^\beta(\tilde{a}_\Lambda)
		=\omega_\Gamma^\beta\Big(\tilde{a}_\Lambda\fint_{SO(3)} (1-e^{\beta\sum_{X\ni x}(I-\hat{\varrho}_x)\bar{\varphi}_X})e^{-\beta\hat{\varrho}_x\psi_x}
		\mathrm{d}\varrho_x\Big)\,.
	\end{align*}
	The latter equation can be regarded as a linear equation in $\underline{\mathsf{X}}$.
	To this avail we expand the exponential in a series converging in $\mathscr{A}_\Gamma$ in the right-hand side:
	\begin{multline*}
		\omega_\Gamma^\beta\Big(\tilde{a}_\Lambda\fint_{SO(3)} (1-e^{\beta\sum_{X\ni x}(I-\hat{\varrho}_x)\bar{\varphi}_X})e^{-\beta\hat{\varrho}_x\psi_x}
		\mathrm{d}\varrho_x\Big)
		\\
		=-\sum_{n\geq 1}
		\frac{\beta^n}{n!}
		\sum_{\substack{X_1,\ldots,X_n\Subset\Gamma\\ x\in X_1\cap\ldots\cap X_n}}
		\omega_\Gamma^\beta\Big(
		\tilde{a}_\Lambda\fint_{SO(3)}
		\Big(\prod_{\ell=1}^n(I-\hat{\varrho}_x)\bar{\varphi}_{X_\ell}\Big)
		e^{-\beta\hat{\varrho}_x\psi_x}\mathrm{d}\varrho_x\Big)
		\\
		=-\sum_{n\geq 1}
		\frac{\beta^n}{n!}
		\sum_{\substack{X_1,\ldots,X_n\Subset\Gamma\\ x\in X_1\cap\ldots\cap X_n}}
		\sum_{X\subseteq S_n}
		\omega_\Gamma^\beta(\tilde{b}_{X\cup\Lambda_n})\,,
	\end{multline*}
	where in the last equality we applied equation \eqref{Eq: decomposition in eta-free elements - improved} where we recall that $S_n=X_1\cup\ldots\cup X_n$, $\Lambda_n:=\Lambda\cap S_n^c$.
	At this stage we may rewrite the latter equation as
	\begin{align}\label{Eq: uniqueness for KMS classical states - linear equation}
		(I-L_\beta)\underline{\omega}_\Gamma^\beta
		=\underline{\delta}\,.
	\end{align}
	Here $\underline{\delta}\in\underline{\mathsf{X}}$ has been defined, mutatis mutandis, as in equation \eqref{Eq: uniqueness for KMS quantum states - linear equation - source term}.
	Similarly, for all $\underline{f}\in\underline{\mathsf{X}}$, we set $(L_\beta\underline{f})_\varnothing=0$ and
	\begin{align*}
		(L_\beta\underline{f})_\Lambda(\tilde{a}_\Lambda)
		:=-\sum_{n\geq 1}
		\frac{\beta^n}{n!}
		\sum_{\substack{X_1,\ldots,X_n\Subset\Gamma\\ x\in X_1\cap\ldots\cap X_n}}
		\sum_{X\subseteq S_n}
		f_{X\cup\Lambda_n}(\tilde{b}_{X\cup\Lambda_n})\,,
	\end{align*}
	for all non-empty $\Lambda\Subset\Gamma$ and $\tilde{a}_\Lambda\in\tilde{\mathscr{A}}_\Lambda$ ---we recall that $x:=\min_{y\in\Lambda}y$.
	
	We will now prove that $\|L_\beta\|_{\mathcal{B}(\underline{\mathsf{X}})}<1$ for $\beta<\tilde{\beta}_{\textsc{u}}$, where $\tilde{\beta}_{\textsc{u}}$ is defined in equation \eqref{Eq: uniqueness for KMS classical states - estimate of critical temperature}.
	This implies that, $I-L_\beta$ is invertible and that \eqref{Eq: uniqueness for KMS classical states - linear equation} has a unique solution in $\underline{\mathsf{X}}$.
	Since $\omega_\Gamma^\beta\mapsto\underline{\omega}_\Gamma^\beta$ is injective and since any $(\beta,\delta_\Gamma)$-KMS classical state induces a solution to equation \eqref{Eq: uniqueness for KMS classical states - linear equation}, it follows that there exists a unique $(\beta,\delta_\Gamma)$-KMS classical state for all $\beta<\tilde{\beta}_{\textsc{u}}$.
	
	We observe that
	\begin{align*}
		|(L_\beta\underline{f})_\Lambda(\tilde{a}_\Lambda)|
		\leq\|\underline{f}\|_{\underline{\mathsf{X}}}\sum_{n\geq 1}
		\frac{\beta^n}{n!}
		\sum_{\substack{X_1,\ldots,X_n\Subset\Gamma\\ x\in X_1\cap\ldots\cap X_n}}
		\sum_{X\subseteq S_n}
		\|\tilde{b}_{X\cup\Lambda_n}\|\,.
	\end{align*}
	Recalling Remark \ref{Rmk: decomposition in eta-free elements - improved} we may provide a useful bound on the norm of each $\tilde{b}_{X\cup\Lambda_n}$.
	In particular we have
	\begin{align*}
		\|\tilde{b}_{X\cup\Lambda_n}\|
		&\leq 2^{|X|}\Big\|
		\tilde{a}_\Lambda
		\fint_{SO(3)}
		\Big(\prod_{\ell=1}^n(I-\hat{\varrho}_x)\bar{\varphi}_{X_\ell}\Big)
		e^{-\beta\hat{\varrho}_x\psi_x}\mathrm{d}\varrho_x
		\Big\|
		\\
		&= 2^{|X|}\Big\|
		\tilde{a}_\Lambda\sum_{A\subseteq\{1,\ldots,n\}}
		(-1)^{n-|A|}
		\Big(\prod_{\ell\in A}\bar{\varphi}_{X_\ell}\Big)
		\eta_x\Big(\prod_{\ell\notin A}\bar{\varphi}_{X_\ell}\Big)
		\Big\|
		\\
		&\leq 2^{|X|}2^n\|\tilde{a}_\Lambda\|
		\prod_{\ell=1}^n\|\varphi_{X_\ell}\|\,.
	\end{align*}
	Summing up we have
	\begin{align*}
		|(L_\beta\underline{f})_\Lambda(\tilde{a}_\Lambda)|
		&\leq\|\underline{f}\|_{\underline{\mathsf{X}}}\sum_{n\geq 1}
		\frac{\beta^n}{n!}
		\sum_{\substack{X_1,\ldots,X_n\Subset\Gamma\\ x\in X_1\cap\ldots\cap X_n}}
		\sum_{X\subseteq S_n}
		\|\tilde{b}_{X\cup\Lambda_n}\|\,.
		\\
		&\leq
		\|\tilde{a}_\Lambda\|\|\underline{f}\|_{\underline{\mathsf{X}}}
		\sum_{n\geq 1}
		\frac{(2\beta)^n}{n!}
		\sum_{\substack{X_1,\ldots,X_n\Subset\Gamma\\ x\in X_1\cap\ldots\cap X_n}}
		\prod_{\ell=1}^n3^{|X_\ell|}\|\bar{\varphi}_{X_\ell}\|
		\\
		&\leq
		\|\tilde{a}_\Lambda\|\|\underline{f}\|_{\underline{\mathsf{X}}}
		\sum_{n\geq 1}
		\frac{(6\beta)^n}{n!}
		\|\bar{\varphi}\|_{\log 3}^n\,.
	\end{align*}
	Overall we find
	\begin{align*}
		\sup_{\Lambda\Subset\Gamma}
		\sup_{\tilde{a}_\Lambda\in\tilde{\mathscr{A}}_\Lambda}
		\frac{|(L_\beta\underline{f})_\Lambda(\tilde{a}_\Lambda)|}{\|\tilde{a}_\Lambda\|}
		\leq
		\|\underline{f}\|_{\underline{\mathsf{X}}}
		\sum_{n\geq 1}
		\frac{1}{n!}(6\beta\|\bar{\varphi}\|_{\log 3})^n
		=\|\underline{f}\|_{\underline{\mathsf{X}}}
		(e^{6\beta\|\bar{\varphi}\|_{\log 3}}-1)\,.
	\end{align*}
	This proves that $L_\beta\underline{f}\in\underline{\mathsf{X}}$, moreover, $L_\beta\in\mathcal{B}(\underline{\mathsf{X}})$ with
	\begin{align*}
		\|L_\beta\|_{\mathcal{B}(\underline{\mathsf{X}})}
		\leq e^{6\beta\|\bar{\varphi}\|_{\log 3}}-1
		<1
		\quad\Leftrightarrow\quad
		\beta<\tilde{\beta}_{\textsc{u}}\,,
	\end{align*}
	as claimed.
\end{proof}

\begin{remark}\label{Rmk: uniqueness of classical KMS state - removal of C1 bound}
    \noindent

    \begin{enumerate}[(i)]
        \item 
    It is worth noting that the proof of Theorem \ref{Thm: uniqueness for KMS classical states} does not, in fact, rely on condition \eqref{Eq: assumption on potential - C1 norm on potentials}.  
    The latter is only needed to ensure that $\delta_\Gamma$ is a well-defined derivation on $\mathscr{A}_\Gamma$, so that it makes sense to consider classical $\delta_\Gamma$-KMS states, \textit{cf.} equation \eqref{Eq: classical KMS condition}.  
    Since $\delta_\Gamma$ is defined by a series involving Poisson brackets, it is natural to require a condition such as \eqref{Eq: assumption on potential - C1 norm on potentials}, which depends on the $C^1$-norm of the potentials.
    However, in the proof of Theorem \ref{Thm: uniqueness for KMS classical states}, the KMS condition is used only through the invariance property \eqref{Eq: invariance of KMS classical states under G action}.
    This property remains meaningful even when the potentials do not satisfy \eqref{Eq: assumption on potential - C1 norm on potentials}.
    From this point of view, if one focuses solely on condition \eqref{Eq: invariance of KMS classical states under G action}, it would suffice to assume $\|\bar{\varphi}\|_{\log 3}<+\infty$, which involves only the $C^*$-norm of the potentials.
    Clearly, there is a caveat: condition \eqref{Eq: invariance of KMS classical states under G action} is not equivalent to the classical KMS condition ---it is necessary, but not sufficient.

    \item
    Similar to the quantum case, Theorem \ref{Thm: uniqueness for KMS classical states} identifies an "optimal" subcritical inverse temperature $\tilde{\beta}_{\textsc{u}}$ for KMS classical states, \textit{cf.} equation \eqref{Eq: uniqueness for KMS classical states - estimate of critical temperature}.
    As a proof of concept, we may compute the actual value of $\tilde{\beta}_{\textsc{u}}$ in a concrete model.
    Specifically, we will consider the \textbf{classical anisotropic Heisenberg model} on $\Gamma=\mathbb{Z}^\nu$:
    \begin{align*}
    	\varphi_\Lambda := \begin{dcases}
    		-J(\delta(s_x^1s_y^1+s_x^2s_y^2)+s_x^3s_y^3) & \Lambda = \{x,y\}\in \Gamma_b\\
    		0 & \text{otherwise}
    	\end{dcases}\,.
    \end{align*} 
    Here for convenience we have taken $J>0$ to be constant, and $(s_x^1,s_x^2,s_x^3)$ are the three coordinate projections on $\mathbb{S}^2$ at site $x$.
    By direct inspection we have
    \begin{align*}
         \|\bar{\varphi}\|_{\log 3}
         =3\sup_{x\in \mathbb{Z}^\nu}
         \sum_{y:\{y,x\}\in\Gamma_b}J\| \delta(s_x^1 s_y^1 + s_x^2 s_y^2) + s^3_x s^3_y\|
         =6J\nu\max{\{|\delta|,1\}}\,,
     \end{align*}
    since  site each has $2\nu$ nearest neighbors.
    Considering equation \eqref{Eq: uniqueness for KMS classical states - estimate of critical temperature}, we find
    \begin{align*}
        \tilde{\beta}_{\textsc{u}} = 
        \frac{1}{18J\nu\max{\{|\delta|,1\}}}\,.
    \end{align*}
    For comparison, we may consider the "optimal" subcritical inverse temperature $\tilde{\beta}_{\textsc{u}}^{\textsc{fv}}$ obtained applying \cite[Prop. 6.39]{Friedli_Velenik_2017}.
    The latter result ensures uniqueness of states on $\mathscr{A}_\Gamma$ fulfilling the DLR condition, which is equivalent to the classical KMS condition in the present setting \cite{Drago_van_de_Ven_2023}.
    From a technical point of view, the uniqueness result presented in \cite[Prop. 6.39]{Friedli_Velenik_2017} is based on \cite[Lem. 6.99]{Friedli_Velenik_2017}, which can be used to estimate the inverse critical temperature.
    Specifically, uniqueness of KMS classical states is granted provided that
    \begin{align*}
        \exists\varepsilon>0\mid
        \sup_{x \in \mathbb{Z}^\nu} \sum_{\Lambda\ni x} 
        \| e^{-\beta \varphi_\Lambda}- 1 \|\, e^{(3+\varepsilon)|\Lambda|}
        \leq 1\,.
    \end{align*}
    For the case at hand we find
    \begin{align*}
        \sup_{x \in \mathbb{Z}^\nu} \sum_{y: \{y,x\}\in\Gamma_b}  \| e^{-\beta \varphi_{\{x,y\}}} - 1 \| e^{2(3+\varepsilon)} 
        = 2\nu e^{6+2\varepsilon}  (e^{\beta J\max{\{|\delta|,1\}}} - 1)\,,
    \end{align*}
    which is smaller than 1 provided that
    \begin{align*}
        \beta
        \leq\tilde{\beta}_{\textsc{u}}^{\textsc{fv}}
        :=\frac{1}{J\max{\{|\delta|,1\}}} \log\Big( 1 + \frac{1}{2\nu \, e^{6 + 2\varepsilon}} \Big)\,.
    \end{align*}
    It follows that
    \begin{align*}
        \frac{\tilde{\beta}_{\textsc{u}}^{\textsc{fv}}}{\tilde{\beta}_{\textsc{u}}}=
        18\nu\log{(1+\frac{1}{2\nu e^{6+2\varepsilon}})}
        \leq\lim_{\nu\to\infty}18\nu\log{(1+\frac{1}{2\nu e^{6+2\varepsilon}})}
        =\frac{9}{e^{6+2\varepsilon}}
        \leq \frac{9}{e^{6}}
        \approx 0.022\,,
    \end{align*}
    for any $\varepsilon\geq 0$.
    As a result, $\tilde{\beta}_{\textsc{u}}^{\textsc{fv}}\leq\tilde{\beta}_{\textsc{u}}$.
    \end{enumerate}
\end{remark}

\subsection{Common subcritical region for quantum and classical systems}
\label{Subsec: common subcritical region for quantum and classical systems}

We close this section by comparing Theorems \ref{Thm: uniqueness for KMS quantum states}-\ref{Thm: uniqueness for KMS classical states} within the common framework of strict deformation quantization and proving Corollary \ref{Cor: absence of classical and quantum phase transitions}.

We consider the quasi-local algebra of classical observables $\mathscr{A}_\Gamma$ defined in Section \ref{Subsec: proof of uniqueness for quantum KMS states}.
On the quantum side, we take the quasi-local algebra of quantum observables $\mathcal{A}_\Gamma$ over the Hilbert spaces $\mathcal{H}_x := \mathbb{C}^{2j+1}$.
Here $j \in \mathbb{Z}_+/2$ represents the spin of the quantum system and will be regarded as a semiclassical parameter, that is, the limit $j \to \infty$ corresponds to the semiclassical regime.

The results of \cite{Drago_Pettinari_Van_de_Ven_2025,Landsman_1998_I} show that the quasi-local algebras $\mathcal{A}_\Gamma$ and $\mathscr{A}_\Gamma$ fit into the framework of strict deformation quantization.
To avoid unnecessary technicalities, we simply note that this entails the existence of a linear map
\[Q_\Gamma \colon \mathscr{A}_\Gamma \to \mathcal{A}_\Gamma\,,\]
with $\|Q_\Gamma\| = 1$.
In fact, one has a family of such maps for various values of $j$, satisfying appropriate consistency conditions as $j \to \infty$.
For our purposes, it suffices to recall that $Q_\Gamma$ can be used to pull-back states $\omega_\Gamma \in S(\mathcal{A}_\Gamma)$ to states on $\mathscr{A}_\Gamma$: Specifically, $\omega_\Gamma\in S(\mathcal{A}_\Gamma)$ implies that $\omega_\Gamma \circ Q_\Gamma\in S(\mathscr{A}_\Gamma)$.
The semiclassical limit is then obtained by considering weak*-limit points, as $j \to \infty$, of sequences of states on $\mathscr{A}_\Gamma$ obtained in this way.

In \cite{Drago_Pettinari_Van_de_Ven_2025,vandeVen_2024} it was shown that the above construction is compatible with the KMS condition.
Specifically, let $\varphi$ be a family of classical potentials on $\mathscr{A}_\Gamma$.
Then one obtains a corresponding family $\Phi$ of quantum potentials on $\mathcal{A}_\Gamma$ by setting
\[
\Phi_\Lambda := Q_\Gamma(\varphi_\Lambda)\,,
\]
for all $\Lambda \Subset \Gamma$.
Given a $(\beta,\tau^\Gamma)$-KMS quantum state $\omega_\Gamma^\beta$ on $\mathcal{A}_\Gamma$ for each value of $j$, one may consider the semiclassical limit of the associated classical states $\omega_\Gamma^\beta \circ Q_\Gamma$.
It was shown in \cite[Thm. 3.5]{Drago_Pettinari_Van_de_Ven_2025} that any weak*-limit point of such a sequence is a $(\beta,\delta_\Gamma)$-KMS classical state.
Moreover, in the case where both the classical and the quantum KMS states are unique (for all values of $j$), the limit $j \to \infty$ can be taken without passing to subsequences.

For these reasons, it is desirable to obtain a sufficient condition on the family of classical potentials $\varphi$ ensuring the uniqueness of both classical KMS states and quantum KMS states associated with the quantized potentials $\Phi$.
This was achieved in \cite[Thm. 4.7, Rmk. 4.9]{Drago_Pettinari_Van_de_Ven_2025}, at the cost of requiring rather strong regularity assumptions on $\varphi$.
From this perspective, Theorems \ref{Thm: uniqueness for KMS quantum states}-\ref{Thm: uniqueness for KMS classical states} provide a more general condition which is spelled out in Corollary \ref{Cor: absence of classical and quantum phase transitions}.

The proof of the latter corollary goes by direct inspection.
Indeed, since $\|Q_\Lambda\|=1$, condition \eqref{Eq: condition for absence of classical and quantum phase transitions} implies \eqref{Eq: uniqueness for KMS quantum states - condition for uniqueness}.
Thus, there exists a unique $(\beta,\tau^\Gamma)$-quantum KMS state on $\mathcal{A}_\Gamma$ ---notice that this applies for all $j\in\mathbb{Z}_+/2$.
Moreover, condition \eqref{Eq: condition for absence of classical and quantum phase transitions} also implies uniqueness of classical KMS states.
Indeed we have
\begin{align*}
	\beta\|\bar{\varphi}\|_{\log 3}
	\leq\beta\|\bar{\varphi}\|_{\varepsilon+\log 3,2\beta}
	<\frac{1}{6}\frac{\varepsilon}{1+e^\varepsilon}
	<\frac{1}{6}\log 2\,,
\end{align*}
provided $\varepsilon>0$ is chosen small enough.
Thus, Theorem \ref{Thm: uniqueness for KMS classical states} applies and there exists a unique $(\beta,\delta_\Gamma)$-KMS classical state on $\mathscr{A}_\Gamma$.
This completes the proof of Corollary \ref{Cor: absence of classical and quantum phase transitions}.

It is important to stress that the proof of Corollary \ref{Cor: absence of classical and quantum phase transitions} is rather simple and strongly benefits from the fact that Theorem \ref{Thm: uniqueness for KMS quantum states} is formulated without assumption \eqref{Eq: assumption on potential - commutation}.
If this were not the case, one would need to verify whether the quantization $Q_\Gamma(\bar{\varphi}_\Lambda)$ of $\bar{\varphi}_\Lambda$ still commutes with the quantization $Q_\Gamma(\psi_x)$ of $\psi_x$.
This is a highly non-trivial requirement.
Moreover, since such commutativity should hold for all $j\in\mathbb{Z}_+/2$ in order to guarantee uniqueness for every semiclassical parameter, the properties of $Q_\Lambda$ (specifically, the Dirac-Groenewold-Rieffel condition) would then imply
\[
	\{\bar{\varphi}_\Lambda,\psi_x\}=0
	\qquad
	\forall\, x\in\Gamma,\;\Lambda\Subset\Gamma\,.
\]
The latter constitutes a necessary condition on the family of classical potentials, which is dispensed within Corollary \ref{Cor: absence of classical and quantum phase transitions}.



\begin{thebibliography}{lll}
	\bibitem{Aizenman_Gallavotti_Goldstein_Lebowitz_1976}
	Aizenman M., Gallavotti G., Goldstein S., Lebowitz J.L.,
	\textit{Stability and equilibrium states of infinite classical systems},
	Commun.Math. Phys. 48, 1-14 
	\href{https://doi.org/10.1007/BF01609407}
	{(1976)}.
	
	\bibitem{Aizenman_Goldstein_Gruber_Lebowitz_Martin_1977}
	Aizenman M., Goldstein S., Gruber C., Lebowitz J.L., Martin P.,
	\textit{On the equivalence between KMS-states and      
        equilibrium states for classical systems},
	Commun. Math. Phys. 53, 209-220
	\href{https://doi.org/10.1007/BF01609847}
	{(1977)}.
	
	\bibitem{Albeverio_Kondratiev_Rockner_Tsikalenko_1997}
	Albeverio S., Kondratiev Y., R\"ockner M., Tsikalenko T.V.,
	\textit{Uniqueness of Gibbs states for quantum lattice systems},
	Probab Theory Relat Fields 108, 193-218
	\href{https://doi.org/10.1007/s004400050107}
	{(1997)}.

%
	
	\bibitem{Araki_1975}
	Araki H.,
	\textit{On uniqueness of KMS states of one-dimensional quantum lattice systems},
	Commun.Math. Phys. 44, 1-7
	\href{https://doi.org/10.1007/BF01609054}
	{(1975)}.

 
	
	
		
	
	\bibitem{Berezin_1975}
	Berezin F. A.,
	\textit{General concept of quantization},
	Commun. Math. Phys. 40, 153-174
	\href{https://doi.org/10.1007/BF01609397}
	{(1975)}.

    
	
	
    
    \bibitem{Bordermann_Hartmann_Romer_Waldmann_98}
    Bordemann M., R\"omer H., Waldmann S.,
    \textit{A Remark on Formal KMS States in Deformation Quantization},
    Lett.Math.Phys. 45 49-61
    \href{https://doi.org/10.1023/A:1007481019610}{(1998)}.
	
	\bibitem{Bratteli_Robinson_97}
	Bratteli O., Robinson D. W.,
	\textit{Operator algebras and quantum statistical mechanics II}
	Springer-Verlag Berlin Heidelberg 
	\href{https://doi.org/10.1007/978-3-662-09089-3}
	{(1981,1997)}.


    

        
	
	
	\bibitem{Dobrushin_1968_1}
	Dobrushin R. L.,
	\textit{The description of a random field by means of conditional probabilities and conditions of its regularity},
	Theor. Prob. Appl. 13, 197-224
	\href{https://doi.org/10.1137/1113026}
	{(1968)}.
	
	\bibitem{Dobrushin_1968_2}
	Dobrushin R. L.,
	\textit{Gibbsian random fields for lattice systems with pairwise interactions},
	Funct. Anal. Appl. 2, 292-301
	\href{https://doi.org/10.1007/BF01075681}
	{(1968)}.
	
	\bibitem{Dobrushin_1968_3}
	Dobrushin R. L.,
	\textit{The problem of uniqueness of a Gibbs random field and the problem of phase transition},
	Funct. Anal. Appl. 2, 302-312
	\href{https://doi.org/10.1007/BF01075682}
	{(1968)}.
	
	\bibitem{Dobrushin_1969}
	Dobrushin R. L.,
	\textit{Gibbsian random fields. The general case},
	Funct Anal Its Appl 3, 22-28
	\href{https://doi.org/10.1007/BF01078271}{(1969)}.
		
	\bibitem{Dobrushin_1970}
	Dobrushin R. L.,
	\textit{Prescribing a system of random variables by conditional distributions},
	Th. Prob. Appl. 17, 582-600
	\href{https://doi.org/10.1137/1115049}
	{(1970)}.
	
	\bibitem{Drago_Pettinari_Van_de_Ven_2025}
	Drago N., Pettinari L., Van de Ven C.J.F.,
	\textit{Classical and quantum KMS states on spin lattice systems},
	Commun. Math. Phys. 406, 163
	\href{https://doi.org/10.1007/s00220-025-05325-2}{(2025)}.
		
	\bibitem{Drago_van_de_Ven_2023}
	Drago N., Van de Ven C.J.F.,
	\textit{DLR-KMS correspondence on lattice spin systems},
	Lett Math Phys 113, 88
	\href{https://doi.org/10.1007/s11005-023-01710-x}{(2023)}.
	
	\bibitem{Drago_Van_de_Ven_2024}
	Drago N., Van de Ven C.J.F.,
	\textit{Strict deformation quantization and local spin interactions},
	Commun. Math. Phys., 405, 14
	\href{ https://doi.org/10.1007/s00220-023-04887-3}{(2024)}.

	\bibitem{Drago_Waldmann_2024}
	Drago N., Waldmann S.,
	\textit{Classical KMS Functionals and Phase Transitions in Poisson Geometry},
	Journal of Symplectic Geometry Vol. 21, no. 5.
	\href{https://dx.doi.org/10.4310/JSG.2023.v21.n5.a3}{(2023)}.


    
 
	\bibitem{Friedli_Velenik_2017}
	Friedli S., Velenik Y.,
	\textit{Statistical mechanics of lattice systems - A concrete mathematical introduction},
	Cambridge University Press
	\href{https://doi.org/10.1017/9781316882603}
	{(2017)}.
	

    \bibitem{Frohlich_Ueltschi_2015}
    Fr\"ohlich J., Ueltschi D.,
    \textit{Some properties of correlations of quantum lattice systems in thermal equilibrium},
    J. Math. Phys. 56, 053302
    \href{https://doi.org/10.1063/1.4921305}{(2015)}.
    
 
	\bibitem{Gallavotti_Pulvirenti_1976}
	Gallavotti G., Pulvirenti M.,
	\textit{Classical KMS Condition and Tomita-Takesaki Theory},
	Comm. Math. Phys. 46, Number 1, 1-9
	\href{https://doi.org/10.1007/BF01610495}
	{(1976)}.
	
	\bibitem{Gallavotti_Verboven_1975}
	Gallavotti G., Verboven E.,
	\textit{On the classical KMS boundary condition},
	Nuov Cim B, 28, 274-286
	\href{https://doi.org/10.1007/BF02722820}
	{(1975)}.
	
	
	
		
	
	\bibitem{Haag_Hugenholtz_Winnik_1967}
	Haag R., Hugenholtz N. M., Winnink M.,
	\textit{On the equilibrium states in quantum statistical mechanics},
	Comm. Math. Phys. 5:215-236
	\href{https://doi.org/10.1007/BF01646342}
	{(1967)}.
	
	\bibitem{Hastings_Koma_2006}
	Hastings M.B., Koma T.,
	\textit{Spectral Gap and Exponential Decay of Correlations},
	Commun. Math. Phys. 265, 781-804
	\href{https://doi.org/10.1007/s00220-006-0030-4}
	{(2006)}.

        
	\bibitem{Israel_1979} 
	Israel R.,
	\textit{Convexity in the theory of lattice gases},
	Princeton Series in Physics, New Jersey
	{(1979)}.
	
	
	
	\bibitem{Kishimoto_1976}
	Kishimoto A.,
	\textit{On uniqueness of KMS states of one-dimensional quantum lattice systems},
	Commun.Math. Phys. 47, 167-170
	\href{https://doi.org/10.1007/BF01608374}
	{(1976)}.
	
    
	\bibitem{Lanford_Ruelle_1969}
	Lanford O.E., Ruelle. D.,
	\textit{Observables at infinity and states with short range correlations in statistical mechanics},
	Comm. Math. Phys., 13: 194-215,
	\href{https://doi.org/10.1007/BF01645487}
	{(1969)}.
	
	\bibitem{Landsman_1998_I}
	Landsman K.,
	\textit{Strict quantization of coadjoint orbits},
	J. Math. Phys. 39, 6372-6383
	\href{https://doi.org/10.1063/1.532644}
	{(1998)}
	
	
	\bibitem{Landsman_2017}
	Landsman K.,
	\textit{Foundations of Quantum Theory: from classical concepts to operators algebras},
	Springer Cham 
	\href{https://doi.org/10.1007/978-3-319-51777-3}
	{(2017)}.
		
	
	\bibitem{Lenci_Rey_Bellet_2005}
	Lenci M., Rey-Bellet L.,
	\textit{Large Deviations in Quantum Lattice Systems: One-Phase Region},
	J Stat Phys 119, 715-746
	\href{https://doi.org/10.1007/s10955-005-3015-3}{(2005)}.
	
	\bibitem{Lieb_Robinson_1972}
	Lieb E.H., Robinson D.W.,
	\textit{The finite group velocity of quantum spin systems},
	Commun.Math. Phys. 28, 251-257
	\href{https://doi.org/10.1007/BF01645779}
	{(1972)}.
	
	
	
	

	\bibitem{Nachtergaele_Ogata_Sims_2006}
	Nachtergaele B., Ogata Y., Sims R.,
	\textit{Propagation of Correlations in Quantum Lattice Systems},
	J Stat Phys 124, 1-13
	\href{https://doi.org/10.1007/s10955-006-9143-6}
	{(2006)}.

    \bibitem{Nachtergaele_Raz_Schlein_Sims_2009}
    Nachtergaele B., Raz H., Schlein B., Sims R.,
    \textit{Lieb-Robinson Bounds for Harmonic and Anharmonic Lattice Systems},
    Commun. Math. Phys. 286 1073-1098
    \href{ https://doi.org/10.1007/s00220-008-0630-2}{(2009)}.
    
    \bibitem{Nachtergaele_Sims_2006}
    Nachtergaele B., Sims R.,
    \textit{Lieb-Robinson Bounds and the Exponential Clustering Theorem}, Commun. Math. Phys. 265, 119-130
    \href{https://doi.org/10.1007/s00220-006-1556-1}
    {(2006)}.
    
    \bibitem{Nachtergaele_Sims_2014}
        Nachtergaele B., Sims R.,
        \textit{On the dynamics of lattice systems with unbounded on-site terms in the Hamiltonian},
        arXiv:1410.8174 [math-ph],
        \href{https://doi.org/10.48550/arXiv.1410.8174}{(2014)}.
    
    \bibitem{Nachtergaele_Schlein_Sims_Starr_Zagrebnov_2010}
    Nachtergaele B., Schlein B., Sims R., Starr S., Zagrebnov V.,
    \textit{On the existence of the dynamics for anharmonic quantum oscillator systems},
    Reviews in Mathematical Physics
    Vol. 22, No. 2 207-231
    \href{https://doi.org/10.1142/S0129055X1000393X}{(2010)}.
    



    
     
     \bibitem{Pusz_Woronowicz_1978}
     Pusz W., Woronowicz S.L.,
     \textit{Passive states and KMS states for general quantum systems},
     Commun.Math. Phys. 58, 273-290
     \href{https://doi.org/10.1007/BF01614224}{(1978)}.
     

 
	\bibitem{Rieffel_94}
	Rieffel M. A.,
	\textit{Quantization and $C^*$-algebras},
	Contemp. Math. 167 67-97
	\href{http://www.ams.org/books/conm/167/conm167-endmatter.pdf}
	{(1994)}.

    
	\bibitem{Ruelle_1967}
	Ruelle D.,
	\textit{A variational formulation of equilibrium statistical mechanics and the Gibbs phase rule},
	Comm. Math. Phys. 5, 324-329
	\href{https://doi.org/10.1007/BF01646446}
	{(1967)}.
	
	
	\bibitem{Sakai_1976}
	Sakai S.,
	\textit{On commutative normal $*$-derivations, II}
	Journal of Functional Analysis, Volume 21, Issue 2,
	\href{https://doi.org/10.1016/0022-1236(76)90078-1}
	{(1976)}.

	
	
	
	
	
	\bibitem{vandeVen_2024}
	van de Ven C.\,J.\,F.,
	\textit{Gibbs states and their classical limit},
	Rev. Math. Phys. 36, 5
	\href{https://doi.org/10.1142/S0129055X24500090}
	{(2024)}.

        
\end{thebibliography}
\end{document}